%% file: Successive.tex
\documentclass[english,onecolumn,draftcls]{IEEEtran}
\usepackage[T1]{fontenc}
\usepackage[latin9]{inputenc}
\usepackage{amsthm}
\usepackage{amsmath}
\usepackage{amssymb}
\usepackage{graphicx}
\usepackage{dsfont}
\usepackage{epstopdf}

\input{preamble}

\makeatletter

\theoremstyle{plain}
\newtheorem{thm}{\protect\theoremname}
\theoremstyle{plain}
\newtheorem{lem}{\protect\lemmaname}
\theoremstyle{plain}

\usepackage{flushend}
\usepackage[noadjust]{cite}
\usepackage{enumerate}
\usepackage{amsmath}
\usepackage{graphicx}
\usepackage{amssymb}
\usepackage{subfigure}
\usepackage{booktabs}
\usepackage{amsfonts}
\usepackage{amssymb}

\makeatother

\usepackage{babel}
\providecommand{\lemmaname}{Lemma} 
\providecommand{\theoremname}{Theorem}
\providecommand{\propositionname}{Proposition}

\begin{document}
\author{Jonathan Scarlett, Alfonso Martinez and Albert Guill\'en i F\`abregas}

\long\def\symbolfootnote[#1]#2{\begingroup\def\thefootnote{\fnsymbol{footnote}}\footnote[#1]{#2}\endgroup}

\title{Mismatched Multi-letter Successive \\
Decoding for the Multiple-Access Channel}
\maketitle
\begin{abstract}
    This paper studies channel coding for the discrete memoryless multiple-access
    channel with a given (possibly suboptimal) decoding rule. A multi-letter successive decoding
    rule depending on an arbitrary non-negative decoding metric
    is considered, and achievable rate regions and error exponents are
    derived both for  the standard MAC (independent codebooks), and for the 
    cognitive MAC (one user knows both messages) with superposition coding.  In the cognitive case,
    the rate region and error exponent are shown to be tight with respect to the ensemble average. The rate regions are 
    compared with those of the commonly-considered decoder that chooses the message pair maximizing the decoding metric, and numerical examples are given for which 
    successive decoding yields a strictly higher sum rate for a given pair 
    of input distributions.
\end{abstract}

\section{Introduction \label{sec:INTRO}}

\symbolfootnote[0]{J.~Scarlett is with the Department of Computer Science and Department of Mathematics, National University of Singapore, 117417. (e-mail: scarlett@comp.nus.edu.sg). 
A.~Martinez is with the Department of Information and Communication Technologies,  
Universitat Pompeu Fabra, 08018 Barcelona, Spain (e-mail: alfonso.martinez@ieee.org). 
A.~Guill\'en i F\`abregas is with the Instituci\'o Catalana de Recerca i Estudis 
Avan\c{c}ats (ICREA), the Department of Information and Communication Technologies,  
Universitat Pompeu Fabra, 08018 Barcelona, Spain, and also with the Department of 
Engineering, University of Cambridge, Cambridge, CB2 1PZ, U.K. (e-mail:  guillen@ieee.org).

This work has been funded in part by the European Research Council under ERC grant agreements 259663 and 725411, by the European Union's 7th Framework Programme under grant agreement 303633 and by the Spanish Ministry of Economy and Competitiveness under grants RYC-2011-08150, TEC2012-38800-C03-03 and TEC2016-78434-C3-1-R.

This work was presented
in part at the 2014 IEEE International Symposium on Information Theory, Honolulu, HI.}

The mismatched decoding problem \cite{Csiszar2,MMRevisited,MacMM} seeks
to characterize the performance of channel coding when the decoding
rule is fixed and possibly suboptimal (e.g., due to channel uncertainty
or implementation constraints). Extensions of this problem to multiuser
settings are not only of interest in their own right, but can also
provide valuable insight into the single-user setting \cite{MacMM,MMSomekh,JournalMU}.
In particular, significant attention has been paid to the mismatched
multiple-access channel (MAC), described as follows.  User $\nu=1,2$ transmits 
a codeword $\xv_{\nu}$ from a codebook $\Cc_{\nu}=\{\xv_{\nu}^{(1)},\cdots,\xv_{\nu}^{(M_{\nu})}\}$,
and the output sequence $\yv$ is generated according to $W^n(\yv|\xv_1,\xv_2) \defeq \prod_{i=1}^{n}W(y_{i}|x_{1,i},x_{2,i})$
for some transition law $W(y|x_1,x_2)$.  The \emph{mismatched} decoder estimates the message pair as
\begin{equation}
    (\hat{\msg}_{1},\hat{\msg}_{2})=\arg\max_{(i,j)}q^{n}(\xv_{1}^{(i)},\xv_{2}^{(j)},\yv),\label{eq:SUC_MaxMetric}
\end{equation}
where $q^{n}(\xv_{1},\xv_{2},\yv)\defeq\prod_{i=1}^{n}q(x_{1,i},x_{2,i},y_{i})$
for some non-negative decoding metric $q(x_{1},x_{2},y)$.
The metric $q(x_{1},x_{2},y) = W(y|x_1,x_2)$ corresponds to optimal
maximum-likelihood (ML) decoding, whereas the introduction of mismatch
can significantly increase the error probability and lead to smaller achievable rate regions \cite{Csiszar2,MacMM}.  Even in the single-user
case, characterizing the capacity with mismatch is a long-standing open problem.

Given that the decoder only knows the metric $q^{n}(\xv_{1}^{(i)},\xv_{2}^{(j)},\yv)$
corresponding to each codeword pair, one may question whether there
exists a decoding rule that provides better performance than the maximum-metric
rule in \eqref{eq:SUC_MaxMetric}, and that is well-motivated from a practical perspective.  The second of these requirements is not redundant; for instance, if the values $\{\log q(x_{1},x_{2},y)\}$ are
rationally independent (i.e., no values can be written as linear combinations
of the others with rational coefficients), then one could consider a highly artificial and impractical decoder that uses these values to infer the joint empirical distribution of $(\xv_{1},\xv_{2},\yv)$, and in turn uses that to implement the maximum-likelihood (ML) rule.  While such a decoder is a function of $\{q^{n}(\xv_{1}^{(i)},\xv_{2}^{(j)},\yv)\}_{i,j}$ and clearly outperforms the maximum-metric rule, it does not bear any practical interest.

There are a variety of well-motivated decoding rules that are of interest beyond maximum-metric, including threshold decoding \cite{Feinstein,MMGeneralFormula}, likelihood decoding \cite{OneShot,PaperLikelihood}, and successive decoding \cite{NetworkBook,GrantRateSplit}.  In this paper, we focus on the latter, and consider the following two-step decoding rule:
\begin{align}
    \hat{\msg}_{1} & =\arg\max_{i}\sum_{j}q^{n}(\xv_{1}^{(i)},\xv_{2}^{(j)},\yv), \label{eq:SUC_Decoder1} \\
    \hat{\msg}_{2} & =\arg\max_{j}q^{n}(\xv_{1}^{(\hat{\msg}_{1})},\xv_{2}^{(j)},\yv). \label{eq:SUC_Decoder2} 
\end{align}
The study of this decoder is of interest for several reasons:
\begin{itemize}
    \item The decoder depends on the exact same quantities as the maximum-metric decoder \eqref{eq:SUC_MaxMetric} (namely, $q^{n}(\xv_{1}^{(i)},\xv_{2}^{(j)},\yv)$ for each $(i,j)$), meaning a comparison of the two rules is in a sense fair.  We will see the successive rule can sometimes achieve rates that are not achieved by the maximum-metric rule (in the random coding regime), which is the first result of this kind for the mismatched MAC.
    \item The first decoding step \eqref{eq:SUC_Decoder1} can
    be considered a mismatched version of the optimal decoding rule for
    (one user of) the interference channel. Hence, as well as giving
    an achievable rate region for the MAC with mismatched successive decoding, our results directly quantify the loss due to mismatch for the interference channel.
    \item More broadly, successive decoding is of significant practical interest for multiple-access scenarios, since it permits the use of single-user codes, as well as linear decoding complexity in the number of users \cite{GrantRateSplit}.  While the specific successive decoder that we consider does not enjoy these practical benefits, it may still serve as an interesting point of comparison for such variants.
\end{itemize}
The rule in \eqref{eq:SUC_Decoder1} is \emph{multi-letter}, in the sense that the objective function does not factorize into a product of $n$ symbols on $(\Xc_1,\Yc)$.  Single-letter successive decoders \cite[Sec.~4.5.1]{NetworkBook} could also potentially be studied from a mismatched decoding perspective by introducing a second decoding metric $q_2(x_1,y)$, but we focus on the above rule depending only on a \emph{single} metric $q(x_1,x_2,y)$.

Under the above definitions of $W$, $q$, $W^n$ and $q^n$, and assuming the corresponding alphabets $\Xc_{1}$, $\Xc_{2}$ and $\Yc$ to be finite, we consider two distinct classes of MACs:
\begin{enumerate}
    \item For the \emph{standard MAC} \cite{MacMM}, encoder $\nu=1,2$ takes as input $\msg_{\nu}$
    equiprobable on $\{1,\cdots,M_{\nu}\}$, and transmits the corresponding
    codeword $\xv_{\nu}^{(\msg_{\nu})}$ from a codebook $\Cc_{\nu}$.
    \item For the \emph{cognitive MAC} \cite{MMSomekh} (or \emph{MAC with degraded message sets} 
    \cite[Ex.~5.18]{NetworkBook}), the messages $\msg_{\nu}$ are still
    equiprobable on $\{1,\cdots,M_{\nu}\}$, but user 2 has access to both messages,
    while user 1 only knows $\msg_1$.  Thus, $\Cc_{1}$ contains codewords
    indexed as $\xv_1^{(i)}$, and $\Cc_2$ contains codewords indexed as $\xv_2^{(i,j)}$.
\end{enumerate}
For each of these, we say that a rate pair $(R_{1},R_{2})$ is achievable if, for all
$\delta>0$, there exist sequences of codebooks $\Cc_{1,n}$
and $\Cc_{2,n}$ with $M_{1}\ge e^{n(R_{1}-\delta)}$ and
$M_{2}\ge e^{n(R_{2}-\delta)}$ respectively, such that the
error probability
\begin{equation}
    \pe \triangleq \PP[(\hat{\msg}_{1},\hat{\msg}_{2})\ne(\msg_{1},\msg_{2})]
\end{equation}
tends to zero under the decoding rule described by \eqref{eq:SUC_Decoder1}--\eqref{eq:SUC_Decoder2}.
Our results will not depend on the method for breaking ties, so for
concreteness, we assume that ties are broken as errors.

For fixed rates $R_1$ and $R_2$, an error exponent $E(R_1,R_2)$ 
is said to be achievable if there exists a sequence of codebooks 
$\Cc_{1,n}$ and $\Cc_{2,n}$ with $M_1 \ge \exp(nR_1)$ 
and $M_2 \ge \exp(nR_2)$ codewords of length $n$ such that
\begin{equation}
    \liminf_{n\to\infty}-\frac{1}{n}\log \pe \ge E(R_1,R_2). \label{eq:SU:AchievableExp}
\end{equation}

Letting $\Ec_{\nu}\defeq\{\hat{\msg}_{\nu}\ne \msg_{\nu}\}$
for $\nu=1,2$, we observe that if $q(x_{1},x_{2},y)=W(y|x_{1},x_{2})$,
then \eqref{eq:SUC_Decoder1} is the decision rule that minimizes
$\PP[\Ec_{1}]$.  Using this observation, we show in Appendix \ref{sub:SUC_PROOF_ML} that
the successive decoder with $q=W$ is guaranteed to achieve the same rate region 
and error exponent as that of optimal non-successive maximum-likelihood decoding.

\subsection{Previous Work and Contributions} \label{sec:previous_work}

The vast majority of previous works on mismatched decoding have focused on
achievability results via random coding, and the only general converse
results are written in terms of non-computable information-spectrum type quantities \cite{MMGeneralFormula}.
For the point-to-point setting with mismatch, the asymptotics of random codes 
with independent codewords are well-understood for the i.i.d.~\cite{Compound},
constant-composition \cite{Hui,Csiszar1,Csiszar2,Merhav} and 
cost-constrained \cite{MMRevisited,JournalSU} ensembles.
Dual expressions and continuous alphabets were studied in \cite{Merhav} and \cite{MMRevisited}.

The mismatched MAC was introduced by Lapidoth \cite{MacMM}, who showed
that $(R_1,R_2)$ is achievable provided that
\begin{align}
    R_{1}       & \le\min_{\substack{\Ptilde_{X_{1}X_{2}Y}\,:\,\Ptilde_{X_{1}}=Q_1,\Ptilde_{X_{2}Y}=P_{X_{2}Y}, \\ \EE_{\Ptilde}[\log q(X_{1},X_{2},Y)]\ge\EE_{P}[\log q(X_{1},X_{2},Y)]}} I_{\Ptilde}(X_{1};X_{2},Y), \label{eq:MAC_R1_LM} \\
    R_{2}       & \le\min_{\substack{\Ptilde_{X_{1}X_{2}Y}\,:\,\Ptilde_{X_{2}}=Q_2,\Ptilde_{X_{1}Y}=P_{X_{1}Y}, \\ \EE_{\Ptilde}[\log q(X_{1},X_{2},Y)]\ge\EE_{P}[\log q(X_{1},X_{2},Y)]}}I_{\Ptilde}(X_{2};X_{1},Y), \label{eq:MAC_R2_LM} \\
    R_{1}+R_{2} & \le\min_{\substack{\Ptilde_{X_{1}X_{2}Y}\,:\,\Ptilde_{X_{1}}=Q_1,\Ptilde_{X_2}=Q_2,\Ptilde_{Y}=P_Y \\
    \EE_{\Ptilde}[\log q(X_{1},X_{2},Y)]\ge\EE_{P}[\log q(X_{1},X_{2},Y)] \\
    I_{\Ptilde}(X_{1};Y)\le R_{1},\, I_{\Ptilde}(X_{2};Y)\le R_{2}}
    }D(\Ptilde_{X_{1}X_{2}Y}\|Q_{1}\times Q_{2}\times\Ptilde_{Y}), \label{eq:MAC_R12_LM}
\end{align}
where $Q_1$ and $Q_2$ are arbitrary input distributions, and $P_{X_1X_2Y} \triangleq Q_1 \times Q_2 \times W$.  The corresponding ensemble-tight error exponent was given by the present authors in \cite{JournalMU}, along with equivalent dual expressions and generalizations to continuous alphabets.  Error exponents were also presented for the MAC with general decoding rules in \cite{NazariLB}, but the results therein are primarily targeted to optimal or universal metrics; in particular, when applied to the mismatched setting, the exponents are not ensemble-tight.

The mismatched cognitive MAC was introduced by Somekh-Baruch \cite{MMSomekh},
who used superposition coding to show that $(R_1,R_2)$ is achievable provided that
\begin{align}
    R_{2}       & \le\min_{\substack{\Ptilde_{X_{1}X_{2}Y}\,:\,\Ptilde_{X_1X_2}=Q_{X_1X_2},\Ptilde_{X_{1}Y}=P_{X_{1}Y}, \\ \EE_{\Ptilde}[\log q(X_{1},X_{2},Y)]\ge\EE_{P}[\log q(X_{1},X_{2},Y)]}}I_{\Ptilde}(X_{2};Y|X_{1}), \label{eq:MAC_R2_LM_C} \\
    R_{1}+R_{2} & \le\min_{\substack{\Ptilde_{X_{1}X_{2}Y}\,:\,\Ptilde_{X_1X_2}=Q_{X_1X_2},\Ptilde_{Y}=P_Y, \\ 
    \EE_{\Ptilde}[\log q(X_{1},X_{2},Y)]\ge\EE_{P}[\log q(X_{1},X_{2},Y)],\, I_{\Ptilde}(X_{1},;Y)\le R_{1}}
    } I_{\Ptilde}(X_1,X_2;Y), \label{eq:MAC_R12_LM_C}
\end{align}
where $Q_{X_1X_2}$ is an arbitrary input distribution, and $P_{X_1X_2Y} \triangleq Q_{X_1X_2} \times W$.
The corresponding ensemble-tight error exponent was also given therein.
Various forms of superposition coding were also studied by the present authors
in \cite{JournalMU}, but with a focus on the single-user channel as opposed
to the cognitive MAC.

Both of the above regions are known to be tight with respect to the ensemble
average for constant-composition random coding, meaning that any looseness
is due to the random-coding ensemble itself, rather than the bounding
techniques used in the analysis \cite{MacMM,MMSomekh}.  This notion of tightness was first
explored in the single-user setting in \cite{Merhav}.  We also note that
the above regions lead to improved achievability bounds for the 
single-user setting \cite{MacMM,MMSomekh}.

The main contributions of this paper are achievable rate regions for both
the standard MAC (Section \ref{sec:SUC_STANDARD}) and cognitive MAC 
(Section \ref{sec:SUC_COGNITIVE}) under the successive decoding rule in
\eqref{eq:SUC_Decoder1}--\eqref{eq:SUC_Decoder2}.  For the cognitive case,
we also provide an ensemble tightness result.  Both regions are numerically compared
to their counterparts for maximum-metric decoding, and in each case, it
is observed that the successive rule can provide a strictly higher sum rate,
though neither the successive nor maximum-metric region is included in the other in general.

A by-product of our analysis is achievable error exponents corresponding to the rate regions.  Our exponent
for the standard MAC is related to that of Etkin \emph{et al.}~\cite{ExponentsIC}
for the interference channel, as both use parallel coding.  Similarly, 
our exponent for the cognitive MAC is related to that of Kaspi and
Merhav \cite{BCExp2}, since both use superposition coding.  Like these
works, we make use of type class enumerators; however, a key difference
is that we avoid applying a Gallager-type bound
in the initial step, and we instead proceed immediately with type-based methods.  


In a work that developed independently of ours, the interference channel perspective 
was pursued in depth in the \emph{matched} case in \cite{WasimIC}, with a focus on error exponents.  The error exponent of \cite{WasimIC} is similar to that derived in the present paper, but also contains an extra maximization term that, at least in principle, could improve the exponent.  Currently, no examples are known where such an improvement is obtained.  Moreover, while the analysis techniques of \cite{WasimIC} extend to the mismatched case, doing so leads to the same achievable rate region as ours; the only potential improvement is in the exponent.  Finally, we note that while our focus is solely on codebooks with independent codewords, error exponents were also given for the Han-Kobayashi construction in \cite{WasimIC}.

Another line of related work studied the achievable rates of polar coding with mismatch \cite{Alsan1,Alsan2,Alsan3,AlsanThesis}, using a computationally efficient successive decoding rule.  A single-letter achievable rate was given, and it was shown that for a given mismatched transition law (i.e., a conditional probability distribution incorrectly used as if it were the true channel), this decoder can sometimes outperform the maximum-metric decoder.  As mentioned above, we make analogous observations in the present paper, albeit for a multiple-access scenario with a very different type of successive decoding.

\subsection{Notation}

Bold symbols are used for vectors (e.g., $\xv$), and the
corresponding $i$-th entry is written using a subscript (e.g., $x_{i}$).
Subscripts are used to denote the distributions corresponding to expectations
and mutual informations (e.g., $\EE_{P}[\cdot]$, $I_{P}(X;Y)$).
The marginals of a joint distribution $P_{XY}$ are denoted by $P_{X}$
and $P_{Y}$. We write $P_{X}=\Ptilde_{X}$ to denote element-wise
equality between two probability distributions on the same alphabet.
The set of all sequences of length $n$ with a given empirical distribution
$P_{X}$ (i.e., type \cite[Ch. 2]{CsiszarBook}) is denoted by $T^{n}(P_{X})$,
and similarly for joint types.
We write $f(n)\doteq g(n)$ if $\lim_{n\to\infty}\frac{1}{n}\log\frac{f(n)}{g(n)}=0$,
and similarly for $\dot{\le}$ and $\dot{\ge}$. We write $[\alpha]^{+}=\max(0,\alpha)$,
and denote the indicator function by $\openone\{\cdot\}$
 
\section{Main Results}

\subsection{Standard MAC} \label{sec:SUC_STANDARD}

Before presenting our main result for the standard MAC, we state the random-coding
distribution that is used in its proof.  For $\nu=1,2$, we fix an input distribution
$Q_{\nu}\in\Pc(\Xc_{\nu})$, and let $Q_{\nu,n}$ be a type with the same support as $Q_{\nu}$
such that $\max_{x_{\nu}}|Q_{\nu,n}(x_{\nu}) - Q_{\nu}(x_{\nu})|\le\frac{1}{n}$.
We set
\begin{equation}
    P_{\Xv_{\nu}}(\xv_{\nu})=\frac{1}{|T^{n}(Q_{\nu,n})|}\openone\big\{\xv_{\nu}\in T^{n}(Q_{\nu,n})\big\},
\end{equation}
and consider codewords $\{\Xv_{\nu}^{(i)}\}_{i=1}^{M_{\nu}}$ that are independently 
distributed according to $P_{\Xv_{\nu}}$.  Thus,
\begin{equation}
    \Big(\{\Xv_{1}^{(i)}\}_{i=1}^{M_{1}},\{\Xv_{2}^{(j)}\}_{i=1}^{M_{2}}\Big)\sim\prod_{i=1}^{M_{1}}P_{\Xv_{1}}(\xv_{1}^{(i)})\prod_{j=1}^{M_{2}}P_{\Xv_{2}}(\xv_{2}^{(j)}). \label{eq:SUC_Distr1}
\end{equation}

Our achievable rate region is written in terms of the functions
\begin{align}
    \Fbar(\Ptilde_{X_{1}X_{2}Y},\Ptilde_{X_{1}X_{2}Y}^{\prime},R_{2}) &\defeq\max\Big\{\EE_{\Ptilde}[\log q(X_{1},X_{2},Y)], \nonumber \\
        & \qquad\quad~\EE_{\Ptilde^{\prime}}[\log q(X_{1},X_{2},Y)]+\big[R_{2}-I_{\Ptilde^{\prime}}(X_{2};X_{1},Y)\big]^{+}\Big\}, \label{eq:SUC_Fupper} \\
    \Funder(P_{X_{1}X_{2}Y},R_{2}) & \defeq\max\bigg\{\EE_{P}[\log q(X_{1},X_{2},Y)], \nonumber \\
        & \hspace*{-1.6cm} \max_{P_{X_{1}X_{2}Y}^{\prime}\in\Tc_{1}^{\prime}(P_{X_{1}X_{2}Y},R_{2})}\EE_{P^{\prime}}[\log q(X_{1},X_{2},Y)]+R_{2}-I_{P^{\prime}}(X_{2};X_{1},Y)\bigg\},\label{eq:SUC_Flower}
\end{align}
where
\begin{align}
    \Tc_{1}^{\prime}(P_{X_1X_2Y},R_{2}) \defeq\Big\{ P_{X_{1}X_{2}Y}^{\prime}\,:\, P_{X_{1}Y}^{\prime}=P_{X_{1}Y},P_{X_{2}}^{\prime}=P_{X_2},I_{P^{\prime}}(X_{2};X_{1},Y)\le R_{2}\Big\}. \label{eq:SUC_SetT1prime} 
\end{align}
We will see in our analysis that $P_{X_1X_2Y}$ corresponds to the joint type of the transmitted codewords and the output sequence, and $\Ptilde_{X_1X_2Y}$
corresponds to the joint type of some incorrect codeword of user 1, the 
transmitted codeword of user 2, and the output sequence.  Moreover,
$P'_{X_1X_2Y}$ and $\Ptilde'_{X_1X_2Y}$ similarly correspond to joint
types, the difference being that the $X_2$ marginal is 
associated with exponentially many sequences in the
summation in \eqref{eq:SUC_Decoder1}.

\begin{thm} \label{thm:SUC_MainResult} 
    For any input distributions $Q_{1}$ and $Q_{2}$, the pair 
    $(R_{1},R_{2})$ is achievable for the standard MAC with the 
    mismatched successive decoding rule
    in \eqref{eq:SUC_Decoder1}--\eqref{eq:SUC_Decoder2} provided that
    \begin{align}
        R_{1} & \le\min_{(\Ptilde_{X_{1}X_{2}Y},\Ptilde_{X_{1}X_{2}Y}^{\prime})\in\Tc_{1}(Q_1 \times Q_2 \times W,R_{2})}I_{\Ptilde}(X_{1};X_{2},Y)+\big[I_{\Ptilde^{\prime}}(X_{2};X_{1},Y)-R_{2}\big]^{+}, \label{eq:SUC_R1} \\
        R_{2} & \le\min_{\Ptilde_{X_{1}X_{2}Y}\in\Tc_{2}(Q_1 \times Q_2 \times W)}I_{\Ptilde}(X_{2};X_{1},Y),\label{eq:SUC_R2}             
    \end{align}
    where
    \begin{align}
        \Tc_{1}(P_{X_1X_2Y},R_{2}) & \defeq\Big\{(\Ptilde_{X_{1}X_{2}Y},\Ptilde_{X_{1}X_{2}Y}^{\prime})\,:\,\Ptilde_{X_{2}Y}=P_{X_{2}Y},\Ptilde_{X_{1}}=P_{X_1},\nonumber \\
        & \hspace*{-4ex}\Ptilde_{X_{1}Y}^{\prime}=\Ptilde_{X_{1}Y},\Ptilde_{X_{2}}^{\prime}=P_{X_2},\Fbar(\Ptilde_{X_{1}X_{2}Y},\Ptilde_{X_{1}X_{2}Y}^{\prime},R_{2})\ge\Funder(P_{X_{1}X_{2}Y},R_{2})\Big\}, \label{eq:SUC_SetT1} \\
        \Tc_{2}(P_{X_1X_2Y}) &\defeq\Big\{\Ptilde_{X_{1}X_{2}Y}\,:\,\Ptilde_{X_{2}}=P_{X_2},\Ptilde_{X_{1}Y}=P_{X_{1}Y}, \nonumber \\
        & \qquad\qquad\qquad\qquad \EE_{\Ptilde}[\log q(X_{1},X_{2},Y)]\ge\EE_{P}[\log q(X_{1},X_{2},Y)]\Big\}.\label{eq:SUC_SetT2}
    \end{align}
\end{thm}
\begin{proof}
    See Section \ref{sec:SUC_MAIN_PROOF}.
\end{proof}
Although the minimization in \eqref{eq:SUC_R1} is a non-convex optimization
problem, it can be cast in terms of convex optimization problems,
thus facilitating its computation.  The details are provided in
Appendix \ref{sec:SUC_CONVEX}. 

While our focus is on achievable rates, the proof of Theorem \ref{thm:SUC_MainResult} 
also provides error exponents.  The exponent corresponding to \eqref{eq:SUC_R2}
is precisely that corresponding to the error event for user 2 with maximum-metric
decoding in \cite[Sec. III]{JournalMU}, and the exponent corresponding to \eqref{eq:SUC_R1} is given by
\begin{equation}
    \min_{P_{X_{1}X_{2}Y}\,:\, P_{X_{1}}=Q_{1},P_{X_{2}}=Q_{2}}D(P_{X_{1}X_{2}Y}\|Q_{1}\times Q_{2}\times W)+\big[\Itilde_1(P_{X_{1}X_{2}Y},R_{2})-R_{1}\big]^{+},\label{eq:SUC_Exponent}
\end{equation}
where $\Itilde_1(P_{X_{1}X_{2}Y},R_{2})$ denotes the right-hand side
of \eqref{eq:SUC_R1} with an arbitrary distribution $P_{X_{1}X_{2}Y}$
in place of $Q_{1}\times Q_{2}\times W$.  As discussed in Section \ref{sec:previous_work}, this exponent is closely related to a parallel work on the error exponent of the interference channel \cite{WasimIC}.

\subsubsection*{Numerical Example}

We consider the MAC with $\Xc_{1}=\Xc_{2}=\{0,1\}$,
$\Yc=\{0,1,2\}$, and
\begin{equation}
    W(y|x_{1},x_{2})=\begin{cases}
    1-2\delta_{x_{1}x_{2}} & y=x_{1}+x_{2}\\
    \delta_{x_{1}x_{2}} & \mathrm{otherwise},
\end{cases}\label{eq:SUC_Example}
\end{equation}
where $\{\delta_{x_{1}x_{2}}\}$ are constants. The mismatched decoder
uses $q(x_{1},x_{2},y)$ of a similar form, but with a fixed value 
$\delta\in\big(0,\frac{1}{3}\big)$ in place of $\{\delta_{x_1x_2}\}$.
While all such choices of $\delta$ are equivalent for maximum-metric 
decoding, this is not true for successive decoding.

We set $\delta_{00}=0.01$, $\delta_{01}=0.1$,
$\delta_{10}=0.01$, $\delta_{11}=0.3$, $\delta=0.15$, 
and $Q_{1}=Q_{2}=(0.5,0.5)$. Figure \ref{fig:SUC_Mismatched}
plots the achievable rates regions of successive decoding (Theorem
\ref{thm:SUC_MainResult}), maximum-metric decoding (\eqref{eq:MAC_R1_LM}--\eqref{eq:MAC_R12_LM}),
and matched decoding (yielding the same region whether successive
or maximum-metric).

\begin{figure} 
    \begin{centering}
        \includegraphics[width=0.42\paperwidth]{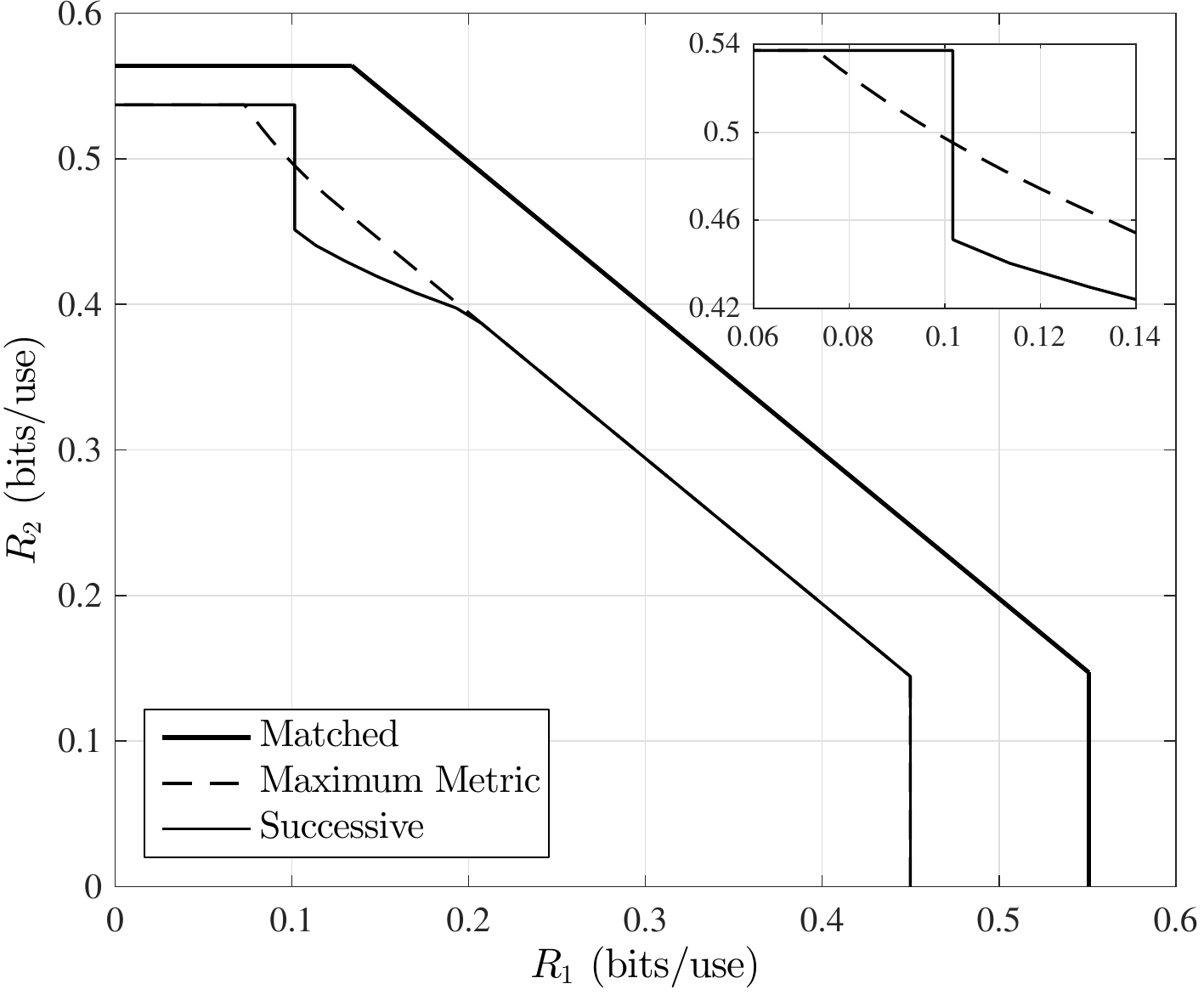}
        \par
    \end{centering}
    
    \caption{Achievable rate regions for the standard MAC given in \eqref{eq:SUC_Example}
        with mismatched successive decoding and mismatched maximum-metric decoding. \label{fig:SUC_Mismatched}}
\end{figure}

Interestingly, neither of the mismatched rate regions is included in the
other, thus suggesting that the two decoding rules are fundamentally
different. For the given input distribution, the sum rate for successive
decoding exceeds that of maximum-metric decoding. Furthermore, upon
taking the convex hull (which is justified by a time sharing argument), 
the region for successive decoding is
strictly larger. While we observed similar behaviors for other choices
of $Q_1$ and $Q_2$, it remains unclear as to whether this is always the
case. Furthermore, while the rate region for maximum-metric decoding
is tight with respect to the ensemble average, it is
unclear whether the same is true of the region given in Theorem \ref{thm:SUC_MainResult}. 

To gain insight into the shape of the achievable rate region for successive decoding, it is instructive to consider the various parts of the region.  When doing so, the reader may wish to note that the condition in \eqref{eq:SUC_R1} can equivalently be expressed as three related conditions holding simultaneously; see Appendix \ref{sec:SUC_CONVEX}, leading to the conditions \eqref{eq:SUC_CondR1b'}, \eqref{eq:SUC_CondR1d'}, and \eqref{eq:SUC_CondR1a''}.  We have the following:
\begin{itemize}
    \item The horizontal line at $R_2 \approx 0.54$ corresponds to the requirement on $R_2$ in \eqref{eq:SUC_R2}, which is identical to the condition in \eqref{eq:MAC_R2_LM} for maximum-metric decoding.
    \item The vertical line at $R_1 \approx 0.45$ also coincides with a condition for maximum-metric decoding, namely, \eqref{eq:MAC_R1_LM}.  It is unsurprising that the two rate regions coincide at $R_2 = 0$, since if user 2 only has one message then the two decoding rules are identical.  For small but positive $R_2$, the rate region boundaries still coincide even though the decoding rules differ, and the successive decoding curve is dominated by condition \eqref{eq:SUC_CondR1a''} in Appendix \ref{sec:SUC_CONVEX}.
    \item The straight diagonal part of the achievable rate region also matches that of maximum-metric decoding.  In this case, the successive decoding curve is dominated by condition \eqref{eq:SUC_CondR1d'} in Appendix \ref{sec:SUC_CONVEX}; the term $\max\{0,I_{\Ptilde'}(X_2;X_1,Y) - R_2\}$ expressed by the $[\cdot]^+$ function is dominated by $I_{\Ptilde'}(X_2;X_1,Y) - R_2$, and the overall condition becomes a sum-rate bound, i.e., an upper bound on $R_1 + R_2$.
    \item In the remaining part of the curve, as $R_1$ gets smaller, the rate region boundary bends downwards, and then suddenly becomes vertical.  In this part, the successive decoding curve is dominated by \eqref{eq:SUC_CondR1b'} in Appendix \ref{sec:SUC_CONVEX}, with $R_2$ being large enough for the term $\max\{0,I_{\Ptilde'}(X_2;X_1,Y) - R_2\}$ to equal zero.  The step-like behavior at $R_1 \approx 0.1$ corresponds to a change in the dominant term of $\Funder$ (see \eqref{eq:SUC_Flower}); in the non-vertical part, the dominant term is $\EE_{\Ptilde}[\log q(X_{1},X_{2},Y)]$, whereas in the vertical part, $R_2$ is large enough for the other term to dominate.
\end{itemize}
It is worth noting that under optimal decoding for the interference channel (which takes the form \eqref{eq:SUC_Decoder1}), it is known that for $R_{1}$ below a certain threshold, $R_{2}$ can be arbitrarily large while still ensuring that user 1's message is estimated correctly \cite{MerhavIC}.  This is in analogy with the step-like behavior observed in Figure \ref{fig:SUC_Mismatched}.

Finally, we note that the mismatched maximum-metric decoding region also has a non-pentagonal and non-convex shape (see the zoomed part of Figure \ref{fig:SUC_Mismatched}), though its deviation from the usual pentagonal shape is milder than the successive decoder in this example.

\subsection{Cognitive MAC} \label{sec:SUC_COGNITIVE}

In this section, we consider the analog of Theorem \ref{thm:SUC_MainResult} for the
cognitive MAC.  Besides being of interest in its own right, this will provide a case where 
ensemble-tightness can be established, and with the numerical results still exhibiting
similar phenomena to those shown in Figure \ref{fig:SUC_Mismatched}.

We again begin by introducing the random coding ensemble.  We fix a joint distribution
$Q_{X_1X_2}\in\Pc(\Xc_1 \times \Xc_2)$, let $Q_{X_1X_2,n}$ be the corresponding closest
joint type in the same way as the previous subsection, and write the resulting marginals
as $Q_{X_1}$, $Q_{X_1,n}$, $Q_{X_2|X_1}$,  $Q_{X_2|X_1,n}$, and so on.  We consider
superposition coding, treating user 1's messages as the ``cloud centers'', and user 2's
messages as the ``satellite codewords''.  More precisely, defining
\begin{align}
    P_{\Xv_1}(\xv_1) &= \frac{1}{|T^{n}(Q_{X_1,n})|}\openone\big\{\xv_1\in T^{n}(Q_{X_1,n})\big\}, \label{eq:COG_PX1} \\
    P_{\Xv_2|\Xv_1}(\xv_2|\xv_1) &= \frac{1}{|T^{n}_{\xv_1}(Q_{X_2|X_1,n})|}\openone\big\{\xv_2\in T^{n}_{\xv_1}(Q_{X_2|X_1,n})\big\}, \label{eq:COG_PX2}
\end{align}
the codewords are distributed as follows:
\begin{equation}
    \bigg\{\Big(\Xv_1^{(i)},\{\Xv_2^{(i,j)}\}_{j=1}^{M_{2}}\Big)\bigg\}_{i=1}^{M_{1}} \,\sim\, \prod_{i=1}^{M_{1}}\bigg(P_{\Xv_1}(\xv_1^{(i)})\prod_{j=1}^{M_{2}}P_{\Xv_2|\Xv_1}(\xv_2^{(i,j)}|\xv_1^{(i)})\bigg). \label{eq:COG_DistrCW}
\end{equation}
For the remaining definitions, we use similar notation to the standard MAC, 
with an additional subscript to avoid confusion.
The analogous quantities to \eqref{eq:SUC_Fupper}--\eqref{eq:SUC_SetT1prime} are
\begin{align}
    \Fbarc(\Ptilde_{X_{1}X_{2}Y}^{\prime},R_{2}) &\defeq\EE_{\Ptilde^{\prime}}[\log q(X_{1},X_{2},Y)]+\big[R_{2}-I_{\Ptilde^{\prime}}(X_{2};Y|X_{1})\big]^{+}, \label{eq:COG_Fupper} \\
    \Funderc(P_{X_{1}X_{2}Y},R_{2}) & \defeq\max\bigg\{\EE_{P}[\log q(X_{1},X_{2},Y)], \nonumber \\
        & \hspace*{-1.6cm} \max_{P_{X_{1}X_{2}Y}^{\prime}\in\Tipc(P_{X_{1}X_{2}Y},R_{2})}\EE_{P^{\prime}}[\log q(X_{1},X_{2},Y)]+R_{2}-I_{P^{\prime}}(X_{2};Y|X_{1})\bigg\},\label{eq:COG_Flower}
\end{align}
where
\begin{align}
    \Tipc(P_{X_1X_2Y},R_{2})  &\defeq\Big\{ P_{X_{1}X_{2}Y}^{\prime}\,:\, P_{X_{1}Y}^{\prime}=P_{X_{1}Y},P_{X_1X_2}^{\prime}=P_{X_1X_2},I_{P^{\prime}}(X_{2};Y|X_{1})\le R_{2}\Big\}. \label{eq:COG_SetT1prime}
\end{align}
Our main result for the cognitive MAC is as follows. 

\begin{thm} \label{thm:SUC_MainResult_C} 
    For any input distribution $Q_{X_1X_2}$, the pair $(R_{1},R_{2})$ 
    is achievable for the cognitive MAC with the mismatched successive decoding rule
    in \eqref{eq:SUC_Decoder1}--\eqref{eq:SUC_Decoder2} provided that
    \begin{align}
        R_{1} & \le\min_{\Ptilde_{X_{1}X_{2}Y}^{\prime} \in\Tic(Q_{X_1X_2} \times W,R_{2})}I_{\Ptilde'}(X_{1};Y)+\big[I_{\Ptilde^{\prime}}(X_{2};Y|X_{1})-R_{2}\big]^{+}\label{eq:COG_R1} \\
        R_{2} & \le\min_{\Ptilde_{X_{1}X_{2}Y}\in\Tiic(Q_{X_1X_2} \times W)}I_{\Ptilde}(X_{2};Y|X_{1}),\label{eq:COG_R2}                      
    \end{align}
    where
    \begin{align}
        \Tic(P_{X_1X_2Y},R_{2}) & \defeq\Big\{\Ptilde_{X_{1}X_{2}Y}^{\prime}\,:\,P_{X_1X_2}^{\prime}=P_{X_1X_2}, \Ptilde'_{Y}=P_{Y}, \Fbarc(\Ptilde_{X_{1}X_{2}Y}^{\prime},R_{2})\ge\Funderc(P_{X_{1}X_{2}Y},R_{2})\Big\}, \label{eq:COG_SetT1} \\
        \Tiic(P_{X_1X_2Y}) &\defeq\Big\{\Ptilde_{X_{1}X_{2}Y}\,:\,\Ptilde_{X_1X_2}=P_{X_1X_2},\Ptilde_{X_{1}Y}=P_{X_{1}Y}, \nonumber \\
        & \qquad\qquad\qquad\qquad \EE_{\Ptilde}[\log q(X_{1},X_{2},Y)]\ge\EE_{P}[\log q(X_{1},X_{2},Y)]\Big\}.\label{eq:COG_SetT2}
    \end{align}
    Conversely, for any rate pair $(R_1,R_2)$ failing to meet both of \eqref{eq:COG_R1}--\eqref{eq:COG_R2},
    the random-coding error probability resulting from \eqref{eq:COG_PX1}--\eqref{eq:COG_DistrCW}
    tends to one as $n\to\infty$.
\end{thm}
\begin{proof}
    See Section \ref{sec:COG_MAIN_PROOF}.
\end{proof}

In Appendix \ref{sec:SUC_CONVEX}, we cast \eqref{eq:COG_R1} in terms of convex 
optimization problems.  Similarly to the previous subsection, the 
exponent corresponding to \eqref{eq:COG_R2} is precisely that corresponding 
to the second user in \cite[Thm.~1]{MMSomekh}, and the exponent corresponding to 
\eqref{eq:COG_R1} is given by
\begin{equation}
    \min_{P_{X_{1}X_{2}Y}\,:\, P_{X_1X_2}=Q_{X_1X_2}}D(P_{X_{1}X_{2}Y}\|Q_{X_1X_2}\times W)+\big[\Ioc(P_{X_{1}X_{2}Y},R_{2})-R_{1}\big]^{+},\label{eq:SUC_Exponent}
\end{equation}
where $\Ioc(P_{X_{1}X_{2}Y},R_{2})$ denotes the right-hand side
of \eqref{eq:COG_R1} with an arbitrary distribution $P_{X_{1}X_{2}Y}$
in place of $Q_{X_1X_2}\times W$.  Similarly to the rate region, the proof
of Theorem \ref{thm:SUC_MainResult_C} shows that these exponents are tight
with respect to the ensemble average (sometimes referred to as \emph{exact
random-coding exponents} \cite{ExactErasure}).

\subsubsection*{Numerical Example}

\begin{figure} 
    \begin{centering}
        \includegraphics[width=0.42\paperwidth]{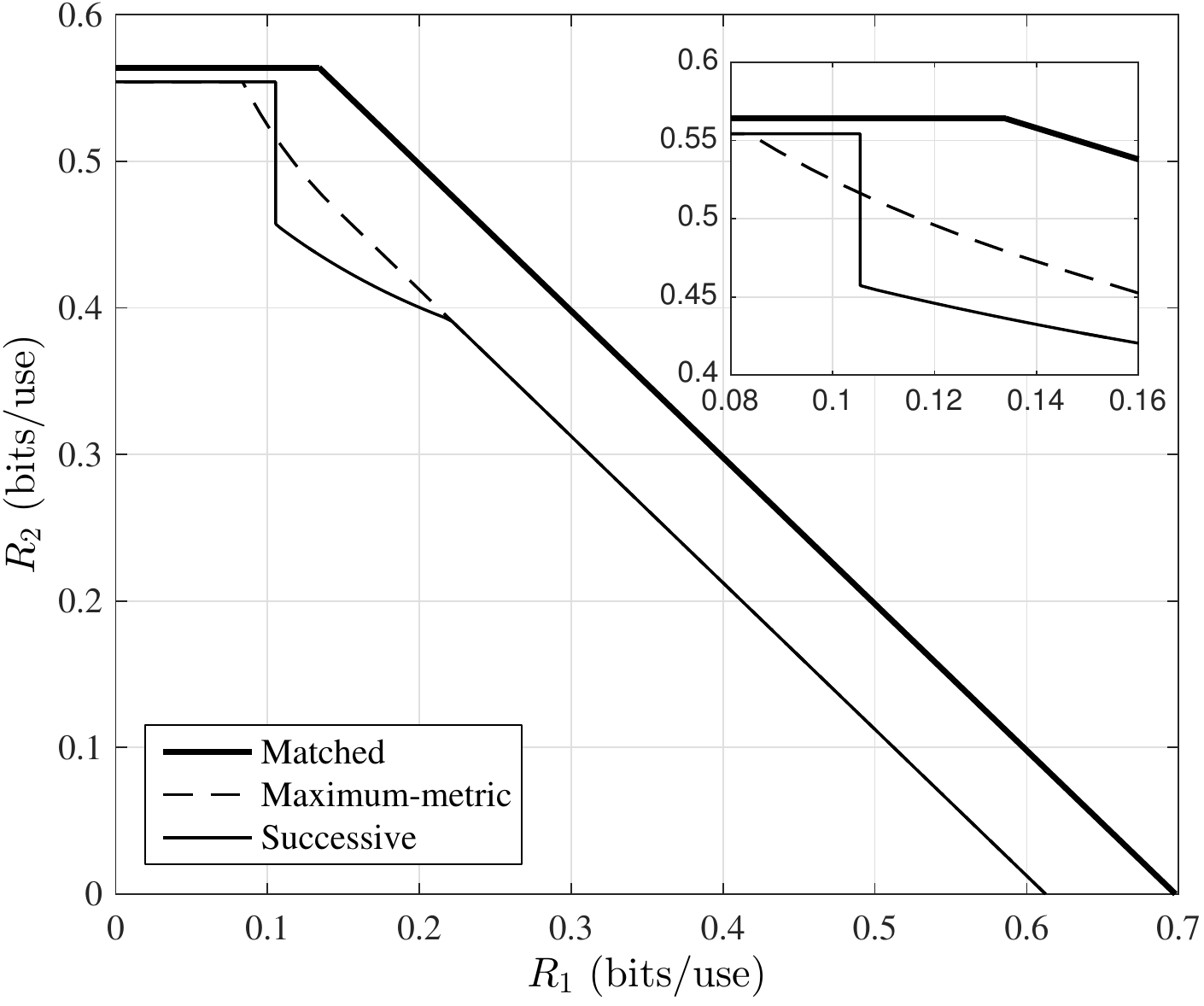}
        \par
    \end{centering} 
    
    \caption{Achievable rate regions for the cognitive MAC given in \eqref{eq:SUC_Example}
        with mismatched successive decoding and mismatched maximum-metric decoding. \label{fig:SUC_Cognitive}}
\end{figure}

We consider again consider the transition law (and the corresponding
decoding metric with a single value of $\delta$) given in 
\eqref{eq:SUC_Example} with $\delta_{00}=0.01$, $\delta_{01}=0.1$,
$\delta_{10}=0.01$, $\delta_{11}=0.3$, $\delta=0.15$, 
and $Q_{X_1X_2} = Q_1 \times Q_2$ with $Q_{1}=Q_{2}=(0.5,0.5)$. 
Figure \ref{fig:SUC_Cognitive} plots the achievable rates regions 
of successive decoding (Theorem \ref{thm:SUC_MainResult_C}), 
maximum-metric decoding (\eqref{eq:MAC_R2_LM_C}--\eqref{eq:MAC_R12_LM_C}), and matched decoding 
(again yielding the same region whether successive or maximum-metric, 
\emph{cf.}~Appendix \ref{sub:SUC_PROOF_ML}).

We see that the behavior of the decoders is completely analogous to
the non-cognitive case observed in Figure \ref{fig:SUC_Mismatched}.
The key difference here is that we know that all three regions are tight with 
respect to the ensemble average.  Thus, we may conclude that the 
somewhat unusual shape of the region for successive decoding is not 
merely an artifact of our analysis, but it is indeed inherent to
the random-coding ensemble and the decoder.  


\section{Proof of Theorem \ref{thm:SUC_MainResult}} \label{sec:SUC_MAIN_PROOF}


The proof of Theorem \ref{thm:SUC_MainResult} is based on the method of type class enumeration (e.g.~see 
\cite{MerhavIC,MerhavErasure,ExactErasure}), and is perhaps most
similar to that of Somekh-Baruch and Merhav \cite{ExactErasure}.

\subsection*{\bf Step 1: Initial Bound}

We assume without loss of generality that $\msg_{1}=\msg_{2}=1$, and we
write $\Xv_{\nu}=\Xv_{\nu}^{(1)}$ and let $\Xvbar_{\nu}$
denote an arbitrary codeword $\Xv_{\nu}^{(j)}$ with $j\ne1$.  Thus,
\begin{equation}
    (\Xv_{1},\Xv_{2},\Yv,\Xvbar_{1},\Xvbar_{2})\sim P_{\Xv_{1}}(\xv_{1})P_{\Xv_{2}}(\xv_{2})W^{n}(\yv|\xv_{1},\xv_{2})P_{\Xv_{1}}(\xvbar_{1})P_{\Xv_{2}}(\xvbar_{2}). \label{eq:SUC_SysDistr}
\end{equation}
We define the following error events:

\smallskip
\begin{minipage}{5in}
\begin{tabbing}
    {\emph{(Type 1)}}~~~ \= $\sum_{j}q^{n}(\Xv_{1}^{(i)},\Xv_{2}^{(j)},\Yv) \ge \sum_{j}q^{n}(\Xv_{1},\Xv_{2}^{(j)},\Yv)$ for some $i \ne 1$; \\
    {\emph{(Type 2)}}~~~ \> $q^{n}(\Xv_{1},\Xv_{2}^{(j)},\Yv) \ge q^{n}(\Xv_{1},\Xv_{2},\Yv)$ for some $j \ne 1$.
\end{tabbing} 
\end{minipage}
\medskip

\noindent Denoting the probabilities of these events by $\peibar$ and $\peiibar$ respectively,
it follows that the overall random-coding error probability $\pebar$ is upper
bounded by $\peibar + \peiibar$.

The analysis of the type-2 error event is precisely that of one of the three 
error types for maximum-metric decoding \cite{MacMM,JournalMU}, yielding the rate
condition in \eqref{eq:SUC_R2}. We thus focus on the type-1 event.
We let $\peibar(\xv_{1},\xv_{2},\yv)$ denote the probability of the type-1 event
conditioned on $(\Xv_{1}^{(1)},\Xv_{2}^{(1)},\Yv)=(\xv_{1},\xv_{2},\yv)$, and we
denote the joint type of $(\xv_{1},\xv_{2},\yv)$ by $P_{X_{1}X_{2}Y}$. 
We write the objective function in \eqref{eq:SUC_Decoder1} as 
\begin{equation}
    \Xi_{\xv_{2}\yv}(\xvbar_{1})\defeq q^{n}(\xvbar_{1},\xv_{2},\yv)+\sum_{j\ne1}q^{n}(\xvbar_{1},\Xv_{2}^{(j)},\yv). \label{eq:SUC_Xi}
\end{equation}
This quantity is random due to the randomness of $\{\Xv_{2}^{(j)}\}$.
The starting point of our analysis is the union bound:
\begin{equation}
    \peibar(\xv_{1},\xv_{2},\yv)\le(M_{1}-1)\PP\big[\Xi_{\xv_{2}\yv}(\Xvbar_{1})\ge\Xi_{\xv_{2}\yv}(\xv_{1})\big].\label{eq:SUC_Union}
\end{equation}
The difficulty in analyzing \eqref{eq:SUC_Union} is that for two different codewords
$\xv_1$ and $\xvbar_1$,  $\Xi_{\xv_{2}\yv}(\xv_{1})$ 
and $\Xi_{\xv_{2}\yv}(\xvbar_{1})$ are not independent, and their joint statistics
are complicated.  We will circumvent this issue by conditioning on high probability events
under which these random quantities can be bounded by deterministic values.

\subsection*{\bf Step 2: An Auxiliary Lemma}

We introduce some additional notation. For a given realization $(\xv_1,\xv_2,\yv)$ of $(\Xv_1,\Xv_2,\Yv)$, we let $\Ptilde_{X_{1}X_{2}Y}$ denote its joint type
and we write $q^{n}(\Ptilde_{X_{1}X_{2}Y}^{\prime})\defeq q^{n}(\xvbar_{1},\xvbar_{2},\yv)$.  In addition, for a general sequence $\xvbar_1$, we define the {\em type enumerator}
\begin{equation}
    N_{\xvbar_{1}\yv}(\Ptilde_{X_{1}X_{2}Y}^{\prime}) = \sum_{j\ne1}\openone\Big\{(\xvbar_{1},\Xv_{2}^{(j)},\yv)\in T^{n}(\Ptilde_{X_{1}X_{2}Y}^{\prime})\Big\},
\end{equation}
which represents the random number of $\Xv_{2}^{(j)}$ $(j\ne1)$ such
that $(\xvbar_{1},\Xv_{2}^{(j)},\yv)\in T^{n}(\Ptilde_{X_{1}X_{2}Y}^{\prime})$.  As we will see below, when $\Xvbar_1 = \xvbar_1$, the quantity $\Xi_{\xv_{2}\yv}(\xvbar_{1})$ can be re-written in terms of $N_{\xvbar_{1}\yv}(\cdot)$, and $\Xi_{\xv_{2}\yv}(\xv_{1})$ can similarly be re-written in terms of $N_{\xv_{1}\yv}(\cdot)$.

The key to replacing random quantities by deterministic ones is to condition on events that hold with probability one approaching faster than exponentially, thus not affecting the exponential behavior of interest.  The following lemma will be used for this purpose, characterizing the behavior of $N_{\xvbar_{1}\yv}(\Ptilde_{X_{1}X_{2}Y}^{\prime})$
for various choices of $R_{2}$ and $\Ptilde_{X_{1}X_{2}Y}^{\prime}$.
The proof can be found in \cite{MerhavIC,ExactErasure}, and is based
on the fact that 
\begin{equation}
    \PP\big[(\xvbar_{1},\Xvbar_{2},\yv)\in T^{n}(\Ptilde_{X_{1}X_{2}Y}^{\prime})\big]\doteq e^{-nI_{\Ptilde^{\prime}}(X_{2};X_{1},Y)}, \label{eq:SUC_Property}
\end{equation}
which is a standard property of types \cite[Ch.~2]{CsiszarBook}. 
\begin{lem}
    \label{lem:SUC_SuperExp} \emph{\cite{MerhavIC,ExactErasure}}
    Fix the pair $(\xvbar_{1},\yv)\in T^{n}(\Ptilde_{X_{1}Y})$,
    a constant $\delta>0$, and a type $\Ptilde_{X_{1}X_{2}Y}^{\prime}\in\Sc_{1,n}^{\prime}(Q_{2,n},\Ptilde_{X_{1}Y})$. 
    \begin{enumerate}
        \item If $R_{2}\ge I_{\Ptilde^{\prime}}(X_{2};X_{1},Y)+\delta$,
        then
        \begin{equation}
            M_{2}e^{-n(I_{\Ptilde^{\prime}}(X_{2};X_{1},Y)+\delta)} \le N_{\xvbar_{1}\yv}(\Ptilde_{X_{1}X_{2}Y}^{\prime}) \le M_{2}e^{-n(I_{\Ptilde^{\prime}}(X_{2};X_{1},Y)-\delta)} \label{eq:SUC_Lem1a}                         
        \end{equation}
        with probability approaching one faster than exponentially.
        \item If $R_{2}<I_{\Ptilde^{\prime}}(X_{2};X_{1},Y)+\delta$,
        then
        \begin{equation}
            N_{\xvbar_{1}\yv}(\Ptilde_{X_{1}X_{2}Y}^{\prime})\le e^{n\,2\delta} \label{eq:SUC_Lem1c}   
        \end{equation}
        with probability approaching one faster than exponentially.
    \end{enumerate}
\end{lem}
Roughly speaking, Lemma \ref{lem:SUC_SuperExp} states that if $R_{2}>I_{\Ptilde^{\prime}}(X_{2};X_{1},Y)$
then the type enumerator is highly concentrated about
its mean, whereas if $R_{2}<I_{\Ptilde^{\prime}}(X_{2};X_{1},Y)$
then the type enumerator takes a subexponential value (possibly zero)
with overwhelming probability.

Given a joint type $\Ptilde_{X_{1}Y}$, define the event
\begin{align}
    \Ac_{\delta}(\Ptilde_{X_{1}Y}) = &\Big\{ \text{\eqref{eq:SUC_Lem1a} holds for all } \Ptilde_{X_{1}X_{2}Y}^{\prime}\in\Sc_{1,n}^{\prime}(Q_{2,n},\Ptilde_{X_{1}Y}) \text{ with } R_{2}\ge I_{\Ptilde^{\prime}}(X_{2};X_{1},Y)+\delta \Big\} \nonumber \\
        \cap\,  &\Big\{ \text{\eqref{eq:SUC_Lem1c} holds for all } \Ptilde_{X_{1}X_{2}Y}^{\prime}\in\Sc_{1,n}^{\prime}(Q_{2,n},\Ptilde_{X_{1}Y}) \text{ with } R_{2}<I_{\Ptilde^{\prime}}(X_{2};X_{1},Y)+\delta \Big\}, \label{eq:A_def}
\end{align}
where
\begin{equation}
    \Sc_{1,n}^{\prime}(Q_{2,n},\Ptilde_{X_{1}Y})\defeq\Big\{\Ptilde_{X_{1}X_{2}Y}^{\prime}\in\Pc_n(\Xc_1\times\Xc_2\times\Yc)\,:\,\Ptilde_{X_{1}Y}^{\prime}=\Ptilde_{X_{1}Y},\Ptilde_{X_{2}}^{\prime}=Q_{2,n}\Big\}, \label{eq:SUC_SetS1n'}
\end{equation}
and where we recall the definition of $Q_{2,n}$ at the start of Section \ref{sec:SUC_STANDARD}.  By Lemma \ref{lem:SUC_SuperExp} and the union bound, $\PP[\Ac_{\delta}(\Ptilde_{X_{1}Y})]\to1$ faster than exponentially.
and hence we can safely condition any event on $\Ac_{\delta}(\Ptilde_{X_{1}Y})$
without changing the exponential behavior of the corresponding probability.
This can be seen by writing the following for any event $\Ec$:
\begin{align}
    \PP[\Ec] &= \PP[\Ec \cap \Ac] + \PP[\Ec \cap \Ac^c] \label{eq:SUC_SuperExp1} \\
        & \le \PP[\Ec \,|\, \Ac] + \PP[\Ac^c],
\end{align}
\begin{align}
    \PP[\Ec] &\ge \PP[\Ec \cap \Ac] \\
        &= (1-\PP[\Ac^c])\PP[\Ec \,|\, \Ac] \\
        &\ge \PP[\Ec \,|\, \Ac] - \PP[\Ac^c]. \label{eq:SUC_SuperExp5}
\end{align}
Using these observations, we will condition on $\Ac_{\delta}$ several times throughout the remainder of the proof.

\subsection*{\bf Step 3: Bound $\Xi_{\xv_{2}\yv}(\xv_{1})$ by a Deterministic Quantity}

From \eqref{eq:SUC_Xi}, we have
\begin{equation}
    \Xi_{\xv_{2}\yv}(\xvbar_{1})=q^{n}(\Ptilde_{X_{1}X_{2}Y})+\sum_{\Ptilde_{X_{1}X_{2}Y}^{\prime}}N_{\xvbar_{1}\yv}(\Ptilde_{X_{1}X_{2}Y}^{\prime})q^{n}(\Ptilde_{X_{1}X_{2}Y}^{\prime}).\label{eq:SUC_TypeMetric}
\end{equation}
Since the codewords are generated independently, $N_{\xvbar_{1}\yv}(\Ptilde_{X_{1}X_{2}Y}^{\prime})$
is binomially distributed with $M_{2}-1$ trials and success probability
$\PP\big[(\xvbar_{1},\Xvbar_{2},\yv)\in T^{n}(\Ptilde_{X_{1}X_{2}Y}^{\prime})\big].$
By construction, we have $N_{\xvbar_{1}\yv}(\Ptilde_{X_{1}X_{2}Y}^{\prime})=0$
unless $\Ptilde_{X_{1}X_{2}Y}^{\prime}\in\Sc_{1,n}^{\prime}(Q_{2,n},\Ptilde_{X_{1}Y})$, where $\Sc_{1,n}^{\prime}$ is defined in \eqref{eq:SUC_SetS1n'}.

Conditioned on $\Ac_{\delta}(P_{X_{1}Y})$, we have the
following:
\begin{align}
      \Xi_{\xv_{2}\yv}(\xv_{1})
      & =q^{n}(P_{X_{1}X_{2}Y})+\sum_{P_{X_{1}X_{2}Y}^{\prime}}N_{\xv_{1}\yv}(P_{X_{1}X_{2}Y}^{\prime})q^{n}(P_{X_{1}X_{2}Y}^{\prime}) \\
      & \ge q^{n}(P_{X_{1}X_{2}Y})+\max_{\substack{P_{X_{1}X_{2}Y}^{\prime}\in\Sc_{1,n}^{\prime}(Q_{2,n},P_{X_{1}Y}) \\
            R_{2}\ge I_{P^{\prime}}(X_{2};X_{1},Y)+\delta}}N_{\xv_{1}\yv}(P_{X_{1}X_{2}Y}^{\prime})q^{n}(P_{X_{1}X_{2}Y}^{\prime})\\
      & \ge q^{n}(P_{X_{1}X_{2}Y})+\max_{\substack{P_{X_{1}X_{2}Y}^{\prime}\in\Sc_{1,n}^{\prime}(Q_{2,n},P_{X_{1}Y}) \\
            R_{2}\ge I_{P^{\prime}}(X_{2};X_{1},Y)+\delta}}M_{2}e^{-n(I_{P^{\prime}}(X_{2};X_{1},Y)+\delta)}q^{n}(P_{X_{1}X_{2}Y}^{\prime})\label{eq:SUC_xiStep3}\\
      & \defeq\Gunder_{\delta,n}(P_{X_{1}X_{2}Y}),\label{eq:SUC_xi_lower}                                                                            
\end{align}
where \eqref{eq:SUC_xiStep3} follows from \eqref{eq:SUC_Lem1a}.
Unlike $\Xi_{\xv_{2}\yv}(\xv_{1})$, the quantity $\Gunder_{\delta,n}(P_{X_{1}X_{2}Y})$ 
is deterministic. Substituting \eqref{eq:SUC_xi_lower} into \eqref{eq:SUC_Union} and
using the fact that $\PP\big[\Ac_{\delta}(\Ptilde_{X_{1}Y})\big]\to1$ 
faster than exponentially, we obtain 
\begin{equation}
    \peibar(\xv_{1},\xv_{2},\yv)\,\,\dot{\le}\, M_{1}\PP\big[\Xi_{\xv_{2}\yv}(\Xvbar_{1})\ge\Gunder_{\delta,n}(P_{X_{1}X_{2}Y})\big].\label{eq:SUC_Union2}
\end{equation}

\subsection*{\bf Step 4: An Expansion Based on Types}

Since the statistics of $\Xi_{\xv_{2}\yv}(\xvbar_{1})$
depend on $\xvbar_{1}$ only through the joint
type of $(\xvbar_{1},\xv_{2},\yv)$,
we can write \eqref{eq:SUC_Union2} as follows: 
\begin{align}
    \peibar(\xv_{1},\xv_{2},\yv) & \,\,\dot{\le}\, M_{1}\sum_{\Ptilde_{X_{1}X_{2}Y}}\PP\big[(\Xvbar_{1},\xv_{2},\yv)\in T^{n}(\Ptilde_{X_{1}X_{2}Y})\big]\PP\big[\Xi_{\xv_{2}\yv}(\xvbar_{1})\ge\Gunder_{\delta,n}(P_{X_{1}X_{2}Y})\big] \\
        & \doteq M_{1}\max_{\Ptilde_{X_{1}X_{2}Y}\in\Sc_{1,n}(Q_{1,n},P_{X_{2}Y})}e^{-nI_{\Ptilde}(X_{1};X_{2},Y)}\PP\big[\Xi_{\xv_{2}\yv}(\xvbar_{1})\ge\Gunder_{\delta,n}(P_{X_{1}X_{2}Y})\big],\label{eq:SUC_Union3}                                     
\end{align}
where $\xvbar_{1}$ denotes an arbitrary sequence such that 
$(\xvbar_{1},\xv_{2},\yv)\in T^{n}(\Ptilde_{X_{1}X_{2}Y})$, and
\begin{equation}
    \Sc_{1,n}(Q_{1,n},P_{X_{2}Y})\defeq\Big\{\Ptilde_{X_{1}X_{2}Y}\in\Pc_n(\Xc_1\times\Xc_2\times\Yc)\,:\,\Ptilde_{X_{1}}=Q_{1,n},\Ptilde_{X_{2}Y}=P_{X_{2}Y}\Big\}.\label{eq:SUC_SetS1n}
\end{equation}
In \eqref{eq:SUC_Union3}, we have used an analogous property to \eqref{eq:SUC_Property}  
and the fact that by construction, the joint type of $(\Xvbar_1,\xv_2,\yv)$ is in
$\Sc_{1,n}(Q_{1,n},P_{X_{2}Y})$ with probability one.

\subsection*{\bf Step 5: Bound $\Xi_{\xv_{2}\yv}(\xvbar_{1})$ by a Deterministic Quantity}

Next, we again use Lemma \ref{lem:SUC_SuperExp} in order to replace
$\Xi_{\xv_{2}\yv}(\xvbar_{1})$
in \eqref{eq:SUC_Union3} by a deterministic quantity. We have from
\eqref{eq:SUC_TypeMetric} that
\begin{equation}
    \Xi_{\xv_{2}\yv}(\xvbar_{1})\le q^{n}(\Ptilde_{X_{1}X_{2}Y})+p_{0}(n)\max_{\Ptilde_{X_{1}X_{2}Y}^{\prime}}N_{\xvbar_{1}\yv}(\Ptilde_{X_{1}X_{2}Y}^{\prime})q^{n}(\Ptilde_{X_{1}X_{2}Y}^{\prime}),\label{eq:SUC_XiBar}
\end{equation} 
where $p_{0}(n)$ is a polynomial corresponding to the total number
of joint types. Substituting \eqref{eq:SUC_XiBar} into \eqref{eq:SUC_Union3}, we obtain
\begin{equation}
    \peibar(\xv_{1},\xv_{2},\yv) \,\,\dot{\le}\,M_{1}\max_{\Ptilde_{X_{1}X_{2}Y}\in\Sc_{1,n}(Q_{1,n},P_{X_{2}Y})}\max_{\Ptilde_{X_{1}X_{2}Y}^{\prime}\in\Sc_{1}^{\prime}(Q_{2,n},\Ptilde_{X_{1}Y})}e^{-nI_{\Ptilde}(X_{1};X_{2},Y)}\PP\big[\Ec_{P,\Ptilde}(\Ptilde_{X_{1}X_{2}Y}^{\prime})\big],\label{eq:SUC_Union4}                                                                                                                                                                                                                                                                                                                                                                                                              
\end{equation}
where
\begin{equation}
    \Ec_{P,\Ptilde}(\Ptilde_{X_{1}X_{2}Y}^{\prime})\defeq\Big\{ q^{n}(\Ptilde_{X_{1}X_{2}Y})+p_{0}(n)N_{\xvbar_{1}\yv}(\Ptilde_{X_{1}X_{2}Y}^{\prime})q^{n}(\Ptilde_{X_{1}X_{2}Y}^{\prime})\ge\Gunder_{\delta,n}(P_{X_{1}X_{2}Y})\Big\}, \label{eq:SUC_ePPtilde}
\end{equation}
and we have used the union bound to take the maximum over $\Ptilde_{X_{1}X_{2}Y}^{\prime}$
outside the probability in \eqref{eq:SUC_Union4}. Continuing, we have
for any $\Ptilde_{X_{1}X_{2}Y}\in\Sc_{1,n}(Q_{1,n},P_{X_{2}Y})$ that
\begin{align}
    \max_{\Ptilde_{X_{1}X_{2}Y}^{\prime}\in\Sc_{1,n}^{\prime}(Q_{2,n},\Ptilde_{X_{1}Y})}\PP\big[\Ec_{P,\Ptilde}(\Ptilde_{X_{1}X_{2}Y}^{\prime})\big]=\max\bigg\{ 
        & \max_{\substack{\Ptilde_{X_{1}X_{2}Y}^{\prime}\in\Sc_{1,n}^{\prime}(Q_{2,n},\Ptilde_{X_{1}Y}) \\ R_{2}\ge I_{\Ptilde^{\prime}}(X_{2};X_{1},Y)+\delta }} \PP\big[\Ec_{P,\Ptilde}(\Ptilde_{X_{1}X_{2}Y}^{\prime})\big], \nonumber \\
        & \max_{\substack{\Ptilde_{X_{1}X_{2}Y}^{\prime}\in\Sc_{1,n}^{\prime}(Q_{2,n},\Ptilde_{X_{1}Y}) \\ R_{2}<I_{\Ptilde^{\prime}}(X_{2};X_{1},Y)+\delta }}\PP\big[\Ec_{P,\Ptilde}(\Ptilde_{X_{1}X_{2}Y}^{\prime})\big]\bigg\}.\label{eq:SUC_TwoMax}                                                                                                                                                                                                    
\end{align}

\subsubsection*{Step 5a -- Simplify the First Term}

For the first term on the right-hand side of \eqref{eq:SUC_TwoMax}, observe that
conditioned on $\Ac_{\delta}(\Ptilde_{X_{1}Y})$
in \eqref{eq:A_def}, we have for $\Ptilde_{X_{1}X_{2}Y}^{\prime}$
satisfying $R_{2}\ge I_{\Ptilde^{\prime}}(X_{2};X_{1},Y)+\delta$ that
\begin{equation}
    N_{\xvbar_{1}\yv}(\Ptilde_{X_{1}X_{2}Y}^{\prime})q^{n}(\Ptilde_{X_{1}X_{2}Y}^{\prime})\le M_{2}e^{-n(I_{\Ptilde^{\prime}}(X_{2};X_{1},Y)-\delta)}q^{n}(\Ptilde_{X_{1}X_{2}Y}^{\prime}),
\end{equation}
where we have used \eqref{eq:SUC_Lem1a}.
Hence, and since $\PP\big[\Ac_{\delta}(\Ptilde_{X_{1}Y})\big]\to1$ 
faster than exponentially, we have 
\begin{equation}
    \PP\big[\Ec_{P,\Ptilde}(\Ptilde_{X_{1}X_{2}Y}^{\prime})\big] \,\,\dot{\le}\,\openone\Big\{ q^{n}(\Ptilde_{X_{1}X_{2}Y})+M_{2}p_{0}(n)e^{-n(I_{\Ptilde^{\prime}}(X_{2};X_{1},Y)-\delta)}q^{n}(\Ptilde_{X_{1}X_{2}Y}^{\prime})\ge\Gunder_{\delta,n}(P_{X_{1}X_{2}Y})\Big\}.\label{eq:SUC_Max1}                                                       
\end{equation}

\subsubsection*{Step 5b -- Simplify the Second Term}

For the second term on the right-hand side of \eqref{eq:SUC_TwoMax}, we define the
event $\Bc\defeq\big\{ N_{\xvbar_{1}\yv}(\Ptilde_{X_{1}X_{2}Y}^{\prime})>0\big\}$,
yielding
\begin{equation}
    \PP[\Bc]\,\,\dot{\le}\, M_{2}e^{-nI_{\Ptilde^{\prime}}(X_{2};X_{1},Y)},\label{eq:SUC_PB}
\end{equation} 
which follows from the union bound and \eqref{eq:SUC_Property}.
Whenever $R_{2}<I_{\Ptilde^{\prime}}(X_{2};X_{1},Y)+\delta$, we have
\begin{align}
      & \PP\big[\Ec_{P,\Ptilde}(\Ptilde_{X_{1}X_{2}Y}^{\prime})\big]\nonumber \\
      & \le\PP\big[\Ec_{P,\Ptilde}(\Ptilde_{X_{1}X_{2}Y}^{\prime})\,\big|\,\Bc^{c}\big]+\PP[\Bc]\PP\big[\Ec_{P,\Ptilde}(\Ptilde_{X_{1}X_{2}Y}^{\prime})\,\big|\,\Bc\big] \\
      & \,\,\dot{\le}\,\openone\big\{ q^{n}(\Ptilde_{X_{1}X_{2}Y})\ge\Gunder_{\delta,n}(P_{X_{1}X_{2}Y})\big\}+M_{2}e^{-nI_{\Ptilde^{\prime}}(X_{2};X_{1},Y)}\PP\big[\Ec_{P,\Ptilde}(\Ptilde_{X_{1}X_{2}Y}^{\prime})\,\big|\,\Bc\big],\label{eq:SUC_Max2_2} \\
      & \,\,\dot{\le}\,\openone\big\{ q^{n}(\Ptilde_{X_{1}X_{2}Y})\ge\Gunder_{\delta,n}(P_{X_{1}X_{2}Y})\big\}\nonumber \\
      & \,\,\,\,+M_{2}e^{-nI_{\Ptilde^{\prime}}(X_{2};X_{1},Y)}\openone\big\{ q^{n}(\Ptilde_{X_{1}X_{2}Y})+p_{0}(n)e^{n\,2\delta}q^{n}(\Ptilde_{X_{1}X_{2}Y}^{\prime})\ge\Gunder_{\delta,n}(P_{X_{1}X_{2}Y})\big\},\label{eq:SUC_Max2_3}                                                   
\end{align}
where \eqref{eq:SUC_Max2_2} follows using \eqref{eq:SUC_PB} and \eqref{eq:SUC_ePPtilde} 
along with the fact that $\Bc^{c}$ implies $N_{\xvbar_{1}\yv}(\Ptilde_{X_{1}X_{2}Y}^{\prime})=0$,
and \eqref{eq:SUC_Max2_3} follows by conditioning on $\Ac_{\delta}(\Ptilde_{X_{1}Y})$
and using \eqref{eq:SUC_Lem1c}.

\subsection*{\bf Step 6: Deduce the Exponent for Fixed $(\xv_1,\xv_2,\yv)$}

Observe that $\Funder(P_{X_{1}X_{2}Y},R_{2})$ in \eqref{eq:SUC_Flower}
equals the exponent of $\Gunder_{\delta,n}$ in \eqref{eq:SUC_xi_lower}
in the limit as $\delta\to0$ and $n\to\infty$. Similarly, the exponents corresponding
to the other quantities appearing in the indicator functions in \eqref{eq:SUC_Max1}
and \eqref{eq:SUC_Max2_3} tend toward the following:
\begin{align}
    \Fbar_{1}(\Ptilde_{X_{1}X_{2}Y},\Ptilde_{X_{1}X_{2}Y}^{\prime},R_{2}) &\defeq \max\Big\{\EE_{\Ptilde}[\log q(X_{1},X_{2},Y)], \nonumber \\ 
        & \qquad\qquad \EE_{\Ptilde^{\prime}}[\log q(X_{1},X_{2},Y)]+R_{2}-I_{\Ptilde^{\prime}}(X_{2};X_{1},Y)\Big\}, \label{eq:SUC_F1upper} \\                                                                                                                                                                                                                                                                                                  
    \Fbar_{2}(\Ptilde_{X_{1}X_{2}Y},\Ptilde_{X_{1}X_{2}Y}^{\prime})       &\defeq \max\Big\{\EE_{\Ptilde}[\log q(X_{1},X_{2},Y)],\EE_{\Ptilde^{\prime}}[\log q(X_{1},X_{2},Y)]\Big\}.\label{eq:SUC_F2upper}
\end{align}
We claim that combining these expressions with \eqref{eq:SUC_Union4}, \eqref{eq:SUC_TwoMax}, 
\eqref{eq:SUC_Max1} and \eqref{eq:SUC_Max2_3} and taking $\delta\to0$
(e.g., analogously to \cite[p.~737]{MMSomekh}, we may set $\delta=n^{-1/2}$),
gives the following:
\begin{align}
      & \peibar(\xv_{1},\xv_{2},\yv)\,\,\dot{\le}\,\max\Bigg\{\max_{\substack{(\Ptilde_{X_{1}X_{2}Y},\Ptilde_{X_{1}X_{2}Y}^{\prime})\in\Tc_{1}^{(1)}(P_{X_{1}X_{2}Y},R_{2})}}M_{1}e^{-nI_{\Ptilde}(X_{1};X_{2},Y)},\nonumber \\
      & \qquad\qquad\max_{\substack{(\Ptilde_{X_{1}X_{2}Y},\Ptilde_{X_{1}X_{2}Y}^{\prime})\in\Tc_{1}^{(2)}(P_{X_{1}X_{2}Y},R_{2})} }M_{1}e^{-nI_{\Ptilde}(X_{1};X_{2},Y)}M_{2}e^{-nI_{\Ptilde^{\prime}}(X_{2};X_{1},Y)}\Bigg\},\label{eq:SUC_peFinal}
\end{align}
where\footnote{Strictly speaking, these sets depend on $(Q_{1},Q_{2})$, but this dependence need not be explicit, since we have $P_{X_{1}}=Q_{1}$ and $P_{X_{2}}=Q_{2}$.}
\begin{multline}
     \Tc_{1}^{(1)}(P_{X_{1}X_{2}Y},R_{2})\defeq\Big\{(\Ptilde_{X_{1}X_{2}Y},\Ptilde_{X_{1}X_{2}Y}^{\prime})\,:\,\Ptilde_{X_{1}X_{2}Y}\in\Sc_{1}(Q_{1},P_{X_{2}Y}),  \\
     \Ptilde_{X_{1}X_{2}Y}^{\prime}\in\Sc_{1}^{\prime}(Q_{2},\Ptilde_{X_{1}Y}),I_{\Ptilde^{\prime}}(X_{2};X_{1},Y)\le R_{2},\Fbar_{1}(\Ptilde_{X_{1}X_{2}Y},\Ptilde_{X_{1}X_{2}Y}^{\prime},R_{2})\ge\Funder(P_{X_{1}X_{2}Y},R_{2})\Big\}, \label{eq:SUC_SetT1_1} 
\end{multline}
\vspace*{-5ex}
\begin{multline}
    \Tc_{1}^{(2)}(P_{X_{1}X_{2}Y},R_{2})\defeq\Big\{(\Ptilde_{X_{1}X_{2}Y},\Ptilde_{X_{1}X_{2}Y}^{\prime})\,:\,\Ptilde_{X_{1}X_{2}Y}\in\Sc_{1}(Q_{1},P_{X_{2}Y}), \\
    \Ptilde_{X_{1}X_{2}Y}^{\prime}\in\Sc_{1}^{\prime}(Q_{2},\Ptilde_{X_{1}Y}),I_{\Ptilde^{\prime}}(X_{2};X_{1},Y)\ge R_{2},\Fbar_{2}(\Ptilde_{X_{1}X_{2}Y},\Ptilde_{X_{1}X_{2}Y}^{\prime})\ge\Funder(P_{X_{1}X_{2}Y},R_{2})\Big\}, \label{eq:SUC_SetT1_2}                                                                                                                                                                                                                                                                                      
\end{multline}
and
\begin{align}
    \Sc_{1}(Q_{1},P_{X_{2}Y})                &\defeq \Big\{\Ptilde_{X_{1}X_{2}Y} \in \Pc(\Xc_1 \times \Xc_2 \times \Yc) \,:\, \Ptilde_{X_{1}}=Q_{1},\Ptilde_{X_{2}Y}=P_{X_{2}Y}\Big\}, \label{eq:SUC_SetS1} \\
    \Sc_{1}^{\prime}(Q_{2},\Ptilde_{X_{1}Y}) &\defeq \Big\{\Ptilde_{X_{1}X_{2}Y}^{\prime} \in \Pc(\Xc_1 \times \Xc_2 \times \Yc) \,:\, \Ptilde_{X_{1}Y}^{\prime}=\Ptilde_{X_{1}Y},\Ptilde_{X_{2}}^{\prime}=Q_{2} \Big\}. \label{eq:SUC_SetS1'}
\end{align}
To see that this is true, we note the following:
\begin{itemize}
    \item For the first term on the right-hand side of \eqref{eq:SUC_peFinal}, the objective function follows from \eqref{eq:SUC_XiBar}, and the additional constraint $\Fbar_{1}(\Ptilde_{X_{1}X_{2}Y},\Ptilde_{X_{1}X_{2}Y}^{\prime},R_{2})\ge\Funder(P_{X_{1}X_{2}Y},R_{2})$ in \eqref{eq:SUC_SetT1_1} follows since the left-hand side in \eqref{eq:SUC_Max1} has exponent $\Fbar_{1}$ and the right-hand side has exponent $\Funder$ by the definition of $\Gunder_{\delta,n}$ in \eqref{eq:SUC_xi_lower}.
    \item For the second term on the right-hand side of \eqref{eq:SUC_peFinal}, the objective function follows from \eqref{eq:SUC_XiBar} and the second term in \eqref{eq:SUC_Max2_3}, and the latter (along with $\Gunder_{\delta,n}$ in \eqref{eq:SUC_xi_lower}) also leads to the final constraint in \eqref{eq:SUC_SetT1_2}.
    \item The first term in \eqref{eq:SUC_Max2_3} is upper bounded by the right-hand side 
    of \eqref{eq:SUC_Max1}, and we already analyzed the latter in order to obtain the first term in \eqref{eq:SUC_peFinal}.  Hence, this term can safely be ignored.
\end{itemize}

\subsection*{\bf Step 7: Deduce the Achievable Rate Region}

By a standard property of types \cite[Ch.~2]{CsiszarBook},
$\PP\big[(\Xv_1,\Xv_2,\Yv)\in T^n(P_{X_1X_2Y})\big]$ decays to zero exponentially
fast when $P_{X_1X_2Y}$ is bounded away from $Q_1 \times Q_2 \times W$.  Therefore,
we can safely substitute $P_{X_{1}X_{2}Y}=Q_{1}\times Q_{2}\times W$ to obtain
the following rate conditions for the first decoding step:
\begin{align}
    R_{1}  & \le\min_{\substack{(\Ptilde_{X_{1}X_{2}Y},\Ptilde_{X_{1}X_{2}Y}^{\prime})\in\Tc_{1}^{(1)}(Q_1 \times Q_2 \times W,R_{2})} } I_{\Ptilde}(X_{1};X_{2},Y),\label{eq:SUC_CondR1}\\
    R_{1}+R_{2} & \le\min_{\substack{(\Ptilde_{X_{1}X_{2}Y},\Ptilde_{X_{1}X_{2}Y}^{\prime})\in\Tc_{1}^{(2)}(Q_1 \times Q_2 \times W,R_{2})} } I_{\Ptilde}(X_{1};X_{2},Y)+I_{\Ptilde^{\prime}}(X_{2};X_{1},Y).\label{eq:SUC_CondR2}
\end{align}
Finally, we claim that \eqref{eq:SUC_CondR1}--\eqref{eq:SUC_CondR2} can be united to obtain \eqref{eq:SUC_R1}.  To see this, we consider two cases:
\begin{itemize}
    \item If $R_{2} > I_{\Ptilde^{\prime}}(X_{2};X_{1},Y)$, then the $[\cdot]^+$ term in \eqref{eq:SUC_R1} equals zero, yielding the objective in \eqref{eq:SUC_CondR1}.  Similarly, in this case, the term $\Fbar$ in \eqref{eq:SUC_Fupper} simplifies to $\Fbar_1$ in \eqref{eq:SUC_F1upper}.
    \item If $R_{2} \le I_{\Ptilde^{\prime}}(X_{2};X_{1},Y)$, then the $[\cdot]^+$ term in \eqref{eq:SUC_R1} equals $I_{\Ptilde^{\prime}}(X_{2};X_{1},Y) - R_2$, yielding the objective in \eqref{eq:SUC_CondR1}.  In this case, the term $\Fbar$ in \eqref{eq:SUC_Fupper} simplifies to $\Fbar_2$ in \eqref{eq:SUC_F2upper}.
\end{itemize}

\section{Proof of Theorem \ref{thm:SUC_MainResult_C}} \label{sec:COG_MAIN_PROOF}
 
The achievability and ensemble tightness proofs for Theorem \ref{thm:SUC_MainResult_C} follow similar steps; to avoid repetition, we focus on the ensemble tightness part.  The achievability part is obtained using exactly the same high-level steps, while occasionally replacing upper bounds by lower bounds as needed via the techniques presented in Section \ref{sec:SUC_MAIN_PROOF}.

\subsection*{\bf Step 1: Initial Bound}

We consider the two error events introduced at the beginning of Section \ref{sec:SUC_MAIN_PROOF}, 
and observe that $\pebar \ge \frac{1}{2}\max\{ \peibar, \peiibar \}$.  The analysis
of $\peiibar$ is precisely that given in \cite[Thm.~1]{MMSomekh}, so we focus on $\peibar$.  

We assume without loss of generality that $\msg_{1}=\msg_{2}=1$, and we
write $\Xv_{\nu}=\Xv_{\nu}^{(1)}$ ($\nu=1,2$), let $\Xv_2^{(j)}$ denote
$\Xv_2^{(1,j)}$, let $\Xvbar_2^{(j)}$ denote 
$\Xv^{(i,j)}$ for some fixed $i\ne1$, and let $(\Xvbar_1,\Xvbar_2)$
denote $(\Xv_1^{(i)},\Xv_2^{(i,j)})$ for some fixed $(i,j)$
with $i \ne 1$.  Thus,
\begin{equation}
    (\Xv_{1},\Xv_{2},\Yv,\Xvbar_{1},\Xvbar_{2})\sim P_{\Xv_{1}}(\xv_{1})P_{\Xv_{2}|\Xv_1}(\xv_{2}|\xv_1)W^{n}(\yv|\xv_{1},\xv_{2})P_{\Xv_{1}}(\xvbar_{1})P_{\Xv_{2}|\Xv_1}(\xvbar_{2}|\xvbar_1). \label{eq:COG_SysDistr}
\end{equation}
Moreover, analogously to \eqref{eq:SUC_Xi}, we define
\begin{gather}
    \Xi_{\xv_{2}\yv}(\xv_{1}) \defeq q^{n}(\xv_{1},\xv_{2},\yv)+\sum_{j\ne1}q^{n}(\xv_{1},\Xv_{2}^{(1,j)},\yv), \label{eq:SUC_Xi_C} \\
    \tilde{\Xi}_{\yv}(\xv_{1}^{(i)}) \defeq \sum_{j}q^{n}(\xv_{1}^{(i)},\Xv_{2}^{(i,j)},\yv). \label{eq:SUC_Xi2_C}
\end{gather}
Note that here we use separate definitions corresponding to $\xv_1$ 
and $\xv_{1}^{(i)}$ ($i\ne1$) since in the cognitive MAC,
each user-1 sequence is associated with a different 
set of user-2 sequences.

Fix a joint type $P_{X_1X_2Y}$ and a triplet $(\xv_1,\xv_2,\yv) \in T^n(P_{X_1X_2Y})$,
and let $\peibar(\xv_1,\xv_2,\yv)$ be the type-1 error probability conditioned on
$(\Xv_1^{(1)},\Xv_2^{(1,1)},\Yv) = (\xv_1,\xv_2,\yv)$; here we assume without
loss of generality that $\msg_1=\msg_2=1$.  We have
\begin{align}
    \peibar(\xv_1,\xv_2,\yv) &= \PP\bigg[\bigcup_{i=2}^{M_1}\Big\{ \tilde{\Xi}_{\yv}(\Xv_1^{(i)}) \ge \Xi_{\xv_2\yv}(\xv_1) \Big\}\bigg] \\
                          &\ge \frac{1}{2}\min\Big\{1,(M_{1}-1)\PP\big[\tilde{\Xi}_{\yv}(\Xvbar_1) \ge \Xi_{\xv_2\yv}(\xv_1)\big]\Big\}, \label{eq:COG_Union}
\end{align}
where \eqref{eq:COG_Union} follows since the truncated union bound is tight to 
within a factor of $\frac{1}{2}$ for independent events \cite[Lemma A.2]{ShulmanThesis}.
Note that this argument fails for the standard MAC; there,
the independence requirement does not hold, so it is unclear whether \eqref{eq:SUC_Union}
is tight upon taking the minimum with $1$.

We now bound the inner probability in \eqref{eq:COG_Union}, which we denote
by $\Phi_1(P_{X_1X_2Y})$.  By similarly defining
\begin{equation}
    \Phi_2(P_{X_1X_2Y},\Ptilde_{X_1Y}) \triangleq \PP\big[\tilde{\Xi}_{\yv}(\Xvbar_1) \ge \Xi_{\xv_2\yv}(\xv_1) \,|\, (\Xvbar_{1},\yv)\in T^{n}(\Ptilde_{X_{1}Y})\big], \label{eq:COG_Phi2}
\end{equation}
we obtain 
\begin{align}
    \Phi_1(P_{X_1X_2Y})
        & \ge\max_{\Ptilde_{X_{1}Y}}\PP\big[(\Xvbar_{1},\yv)\in T^{n}(\Ptilde_{X_{1}Y})\big] \Phi_2(P_{X_1X_2Y},\Ptilde_{X_1Y}) \\
        & \doteq\max_{\Ptilde_{X_{1}Y}\,:\,\Ptilde_{X_1}=Q_{X_1},\Ptilde_{Y}=P_Y}e^{-nI_{\Ptilde}(X_{1};Y)} \Phi_2(P_{X_1X_2Y},\Ptilde_{X_1Y}), \label{eq:COG_PartialTypeExpr}                              
\end{align}
where \eqref{eq:COG_PartialTypeExpr} is a standard property of types \cite[Ch.~2]{CsiszarBook}.
We proceed by bounding $\Phi_2$; to do so, we let $\xvbar_1$ be an arbitrary sequence such
that $(\xvbar_1,\yv) \in T^n(\Ptilde_{X_1Y})$. By symmetry, any such sequence may be considered.

\subsection*{\bf Step 2: Type Class Enumerators}

We write each metric $\Xi_{\xv_2\yv}$ in terms of type class enumerators.  Specifically, 
again writing $q^n(P_{X_1X_2Y})$ to denote the $n$-fold product metric for a given joint type, we 
note the following analogs of \eqref{eq:SUC_TypeMetric}: 
\begin{gather}
    \Xi_{\xv_2\yv}(\xv_1) =q^{n}(P_{X_{1}X_{2}Y})+\sum_{P_{X_{1}X_{2}Y}^{\prime}} \Xi_{\yv}(\xv_1,P_{X_{1}X_{2}Y}^{\prime}),   \\
    \tilde{\Xi}_{\yv}(\xvbar_1) =\sum_{\Ptilde_{X_{1}X_{2}Y}^{\prime}}\tilde{\Xi}_{\yv}(\xvbar_1,\Ptilde_{X_{1}X_{2}Y}^{\prime}), 
\end{gather}
where
\begin{align}
    \Xi_{\yv}(\xv_1,P_{X_{1}X_{2}Y}^{\prime}) & \defeq N_{\xv_{1}\yv}(P_{X_{1}X_{2}Y}^{\prime})q^{n}(P_{X_{1}X_{2}Y}^{\prime}), \\
    \tilde{\Xi}_{\yv}(\xvbar_1,\Ptilde_{X_{1}X_{2}Y}^{\prime}) & \defeq \tilde{N}_{\xvbar_{1}\yv}(\Ptilde_{X_{1}X_{2}Y}^{\prime})q^{n}(\Ptilde_{X_{1}X_{2}Y}^{\prime}), 
\end{align}
and
\begin{align}
    N_{\xv_{1}\yv}(P_{X_{1}X_{2}Y}^{\prime})\defeq\sum_{j\ne1}\openone\Big\{(\xv_{1},\Xv_{2}^{(j)},\yv)\in T^{n}(P_{X_{1}X_{2}Y}^{\prime})\Big\}, \label{eq:COG_TypeEnum1} \\
    \tilde{N}_{\xvbar_{1}\yv}(\Ptilde_{X_{1}X_{2}Y}^{\prime})\defeq\sum_{j}\openone\Big\{(\xvbar_{1},\Xvbar_{2}^{(j)},\yv)\in T^{n}(\Ptilde_{X_{1}X_{2}Y}^{\prime})\Big\}.
\end{align}
Note the minor differences in these definitions compared to those for the standard MAC,
resulting from the differing codebook structure associated with superposition coding.
Using these definitions, we can bound \eqref{eq:COG_Phi2} as follows:
\begin{align}
    \Phi_2(P_{X_1X_2Y},\Ptilde_{X_1Y}) 
    &= \PP\bigg[\sum_{\Ptilde_{X_{1}X_{2}Y}^{\prime}}\tilde{\Xi}_{\yv}(\xvbar_1,\Ptilde_{X_{1}X_{2}Y}^{\prime})\ge q^{n}(P_{X_{1}X_{2}Y})+\sum_{P_{X_{1}X_{2}Y}^{\prime}}\Xi_{\yv}(\xv_1,P_{X_{1}X_{2}Y}^{\prime})\bigg]\\
    &\ge \PP\bigg[\max_{\Ptilde_{X_{1}X_{2}Y}^{\prime}}\tilde{\Xi}_{\yv}(\xvbar_1,\Ptilde_{X_{1}X_{2}Y}^{\prime})\ge q^{n}(P_{X_{1}X_{2}Y})+p_{0}(n)\max_{P_{X_{1}X_{2}Y}^{\prime}}\Xi_{\yv}(\xv_1,P_{X_{1}X_{2}Y}^{\prime})\bigg]\\
    &\ge \max_{\Ptilde_{X_{1}X_{2}Y}^{\prime}}\PP\bigg[\tilde{\Xi}_{\yv}(\xvbar_1,\Ptilde_{X_{1}X_{2}Y}^{\prime})\ge q^{n}(P_{X_{1}X_{2}Y})+p_{0}(n)\max_{P_{X_{1}X_{2}Y}^{\prime}}\Xi_{\yv}(\xv_1,P_{X_{1}X_{2}Y}^{\prime})\bigg] \label{eq:TypeExpr} \\
    &\triangleq \max_{\Ptilde_{X_{1}X_{2}Y}^{\prime}} \Phi_3(P_{X_1X_2Y},\Ptilde_{X_1Y},\Ptilde'_{X_1X_2Y}), \label{eq:COG_DefPhi3}
\end{align}
where $p_0(n)$ is a polynomial corresponding to the number of joint types.

\subsection*{\bf Step 3: An Auxiliary Lemma}

We define the sets
\begin{gather}
    \Sicn(Q_{X_1,n},P_{Y})\defeq\Big\{\Ptilde_{X_{1}Y}\in\Pc_n(\Xc_1\times\Yc)\,:\,\Ptilde_{X_1}=Q_{X_1,n},\Ptilde_{Y}=P_{Y}\Big\}, \label{eq:COG_SetS1n} \\
    \Sicn^{\prime}(Q_{X_1X_2,n},\Ptilde_{X_{1}Y})\defeq\Big\{\Ptilde_{X_{1}X_{2}Y}^{\prime}\in\Pc_n(\Xc_1\times\Xc_2\times\Yc)\,:\,\Ptilde_{X_{1}Y}^{\prime}=\Ptilde_{X_{1}Y},\Ptilde_{X_1X_2}=Q_{X_1X_2,n}\Big\}.\label{eq:COG_SetS1n'}
\end{gather}
The following lemma provides analogous properties to Lemma \ref{lem:SUC_SuperExp}
for joint types within $\Sicn'$, with suitable modifications to handle the fact that 
we are proving ensemble tightness rather than achievability.  It is based
on the fact that $N_{\xvbar_{1}\yv}(\Ptilde_{X_{1}X_{2}Y}^{\prime})$ has a 
binomial distribution with success probability
$\PP[(\xvbar_{1},\Xvbar_2,\yv) \in T^n(\Ptilde_{X_{1}X_{2}Y}^{\prime}) \,|\, \Xvbar_1=\xvbar_1] \doteq e^{-nI_{\Ptilde'}(X_2;Y|X_1)}$
by \eqref{eq:COG_PX2}.

\begin{lem} \label{lem:COG_SuperExp}
    Fix a joint type $\Ptilde_{X_1Y}$ and a pair $(\xvbar_1,\yv) \in T^n(\Ptilde_{X_1Y})$.
    For any joint type $\Ptilde'_{X_1X_2Y} \in \Sc'_{1,n}(Q_{X_1X_2,n},\Ptilde_{X_{2}Y})$
    and constant $\delta >0$, the type enumerator 
    $N_{\xvbar_{1}\yv}(\Ptilde_{X_{1}X_{2}Y}^{\prime})$ satisfies the following:
    \begin{enumerate}
        \item If $R_{2}\ge I_{\Ptilde^{\prime}}(X_{2};Y|X_{1})-\delta$, then
        $N_{\xvbar_{1}\yv}(\Ptilde_{X_{1}X_{2}Y}^{\prime})\le M_{2}e^{-n(I_{\Ptilde^{\prime}}(X_{2};Y|X_{1})-2\delta)}$
        with probability approaching one faster than exponentially.
        \item If $R_{2}\ge I_{\Ptilde^{\prime}}(X_{2};Y|X_{1})+\delta$, then
        $N_{\xvbar_{1}\yv}(\Ptilde_{X_{1}X_{2}Y}^{\prime})\ge M_{2}e^{-n(I_{\Ptilde^{\prime}}(X_{2};Y|X_{1})+\delta)}$
        with probability approaching one faster than exponentially.
        \item If $R_{2}\le I_{\Ptilde^{\prime}}(X_{2};Y|X_{1})-\delta$, then
        \begin{enumerate}
            \item $N_{\xvbar_{1}\yv}(\Ptilde_{X_{1}X_{2}Y}^{\prime})\le e^{n\delta}$
            with probability approaching one faster than exponentially;
            \item $\PP\big[N_{\xvbar_{1}\yv}(\Ptilde_{X_{1}X_{2}Y}^{\prime})>0\big]\doteq M_{2}e^{-nI_{\Ptilde^{\prime}}(X_{2};Y|X_{1})}$.
        \end{enumerate}
    \end{enumerate}
    Moreover, the analogous properties hold for the type enumerator 
    $N_{\xv_{1}\yv}(P_{X_{1}X_{2}Y}^{\prime})$ and any joint types
    $P_{X_1Y}$ (with $P_{X_1} = Q_{X_1,n}$) and 
    $\Ptilde'_{X_1X_2Y} \in \Sc_{1,n}^{\prime}(Q_{X_1X_2,n},P_{X_{1}Y})$.
\end{lem}
\begin{proof}
    Parts 1, 2 and 3a are proved in the same way as Lemma \ref{lem:SUC_SuperExp};
    we omit the details to avoid repetition with \cite{MerhavIC,ExactErasure}.
    Part 3b follows by writing the probability that $N_{\xvbar_{1}\yv} > 0$ as 
    a union of the $M_1-1$ events in \eqref{eq:COG_TypeEnum1} holding, and
    using the fact that the truncated union bound is tight to within a factor
    of $\frac{1}{2}$ \cite[Lemma A.2]{ShulmanThesis}.  The truncation need not
    explicitly be included, since the assumption of part 3 implies that 
    $M_{2}e^{-nI_{\Ptilde^{\prime}}(X_{2};Y|X_{1})} \to 0$.
\end{proof}

Given a joint type $P_{X_{2}Y}$ (respectively, $\Ptilde_{X_1Y}$), let $\Ac_{\delta}(\Ptilde_{X_{1}Y})$
(respectively, $\tilde{\Ac}_{\delta}(\Ptilde_{X_{1}Y})$) denote the union of the 
high-probability events in Lemma \ref{lem:COG_SuperExp} (in parts 1, 2 and 3a)
taken over all $P_{X_{1}X_{2}Y}^{\prime}\in\Sc_{1,n}(Q_{X_1X_2},P_{X_2Y})$
(respectively, $\Ptilde_{X_{1}X_{2}Y}^{\prime}\in\Sc_{1,n}^{\prime}(Q_{X_1X_2},\Ptilde_{X_{1}Y})$).
By the union bound, the probability of these events tends to one faster than exponentially,
and hence we can safely condition any event accordingly
without changing the exponential behavior of the corresponding probability
(see \eqref{eq:SUC_SuperExp1}--\eqref{eq:SUC_SuperExp5}).

\subsection*{\bf Step 4: Bound $\Xi_{\yv}(\xv_1,P_{X_{1}X_{2}Y}^{\prime})$ by a Deterministic Quantity}

We first deal with $\Xi_{\yv}(\xv_1,P_{X_{1}X_{2}Y}^{\prime})$ in \eqref{eq:TypeExpr}.
Defining the event
\begin{equation}
    \Bc_{\delta}\defeq\Big\{ N_{\xv_{1}\yv}(P_{X_{1}X_{2}Y}^{\prime})=0\text{ for all }P_{X_{1}X_{2}Y}^{\prime}\text{ such that }R_{2}\le I_{\Ptilde^{\prime}}(X_{2};Y|X_{1})-\delta\Big\}, \label{eq:B_delta}
\end{equation}
we immediately obtain from Property 3b in Lemma \ref{lem:COG_SuperExp} that 
$\PP\big[\Bc_{\delta}^{c}\big]\,\,\dot{\le}\, e^{-n\delta}\to0$, and hence
\begin{align}
      \Phi_3(P_{X_1X_2Y},\Ptilde_{X_1Y},\Ptilde'_{X_1X_2Y})
      & \ge\PP\bigg[\tilde{\Xi}_{\yv}(\xvbar_1,\Ptilde_{X_{1}X_{2}Y}^{\prime})\ge q^{n}(P_{X_{1}X_{2}Y})+p_{0}(n)\max_{P_{X_{1}X_{2}Y}^{\prime}}\Xi_{\yv}(\xv_1,P_{X_{1}X_{2}Y}^{\prime})\,\cap\,\Bc_{\delta}\bigg] \\
      & \doteq\PP\bigg[\tilde{\Xi}_{\yv}(\xvbar_1,\Ptilde_{X_{1}X_{2}Y}^{\prime})\ge q^{n}(P_{X_{1}X_{2}Y})+p_{0}(n)\max_{P_{X_{1}X_{2}Y}^{\prime}}\Xi_{\yv}(\xv_1,P_{X_{1}X_{2}Y}^{\prime})\,\Big|\,\Bc_{\delta}\bigg].     
\end{align}
Next, conditioned on both $\Bc_{\delta}$ and the events in Lemma \ref{lem:COG_SuperExp}, we have
\begin{align}
      & q^{n}(P_{X_{1}X_{2}Y})+p_{0}(n)\max_{P_{X_{1}X_{2}Y}^{\prime}}\Xi_{\yv}(\xv_1,P_{X_{1}X_{2}Y}^{\prime})  \\
      & \qquad=q^{n}(P_{X_{1}X_{2}Y})+p_{0}(n)\max_{\substack{P_{X_{1}X_{2}Y}^{\prime}\in\Sicn^{\prime}(Q_{X_1X_2,n},\Ptilde_{X_{1}Y})\,:\, \\ R_{2}\ge I_{P^{\prime}}(X_{2};Y|X_{1})-\delta}}\Xi_{\yv}(\xv_1,P_{X_{1}X_{2}Y}^{\prime}) \\
      & \qquad\le q^{n}(P_{X_{1}X_{2}Y})+p_{0}(n)\max_{\substack{P_{X_{1}X_{2}Y}^{\prime}\in\Sicn^{\prime}(Q_{X_1X_2,n},\Ptilde_{X_{1}Y})\,:\, \\ R_{2}\ge I_{P^{\prime}}(X_{2};Y|X_{1})-\delta}}M_{2}e^{-n(I_{\Ptilde^{\prime}}(X_{2};Y|X_{1})-2\delta)}q^{n}(P_{X_{1}X_{2}Y}^{\prime}) \label{eq:Cond_Bd_S3} \\
      & \qquad\defeq\Gbar_{\delta,n}(P_{X_{1}X_{2}Y}),  \label{eq:COG_Gbar}
\end{align}
where in \eqref{eq:Cond_Bd_S3} we used part 1 of Lemma \ref{lem:COG_SuperExp}.
It follows that
\begin{equation}
    \Phi_3(P_{X_1X_2Y},\Ptilde_{X_1Y},\Ptilde'_{X_1X_2Y})\,\,\dot{\ge}\,\PP\big[\tilde{\Xi}_{\yv}(\xvbar_1,\Ptilde_{X_{1}X_{2}Y}^{\prime})\ge\Gbar_{\delta,n}(P_{X_{1}X_{2}Y})\big],\label{eq:COG_HandledP}
\end{equation}
where the conditioning on $\Bc_{\delta}$ has been removed
since it is independent of the statistics of $\tilde{\Xi}_{\yv}(\xvbar_1,\Ptilde_{X_{1}X_{2}Y}^{\prime})$.

\subsection*{\bf Step 5: Bound $\Xi_{\xv_{2}\yv}(\xvbar_{1})$ by a Deterministic Quantity}

We now deal with $\tilde{\Xi}_{\yv}(\xvbar_1,\Ptilde_{X_{1}X_{2}Y}^{\prime})$.
Substituting \eqref{eq:COG_HandledP} into \eqref{eq:COG_DefPhi3} 
and constraining the maximization in two different ways, we obtain
\begin{align}
      \Phi_2(P_{X_1X_2Y},\Ptilde_{X_1Y})
      &\,\,\dot{\ge}\,\max\Bigg\{\max_{\substack{\Ptilde_{X_{1}X_{2}Y}^{\prime}\in\Sicn^{\prime}(Q_{X_1X_2,n},\Ptilde_{X_{1}Y})\,:\, \\ R_{2}\ge I_{\Ptilde^{\prime}}(X_{2};Y|X_{1})+\delta}}\PP\bigg[\tilde{\Xi}_{\yv}(\xvbar_1,\Ptilde_{X_{1}X_{2}Y}^{\prime})\ge\Gbar_{\delta,n}(P_{X_{1}X_{2}Y})\bigg],\nonumber \\
      &\qquad\quad\max_{\substack{\Ptilde_{X_{1}X_{2}Y}^{\prime}\in\Sicn^{\prime}(Q_{X_1X_2,n},\Ptilde_{X_{1}Y})\,:\, \\ R_{2}\le I_{\Ptilde^{\prime}}(X_{2};Y|X_{1})-\delta}}\PP\bigg[\tilde{\Xi}_{\yv}(\xvbar_1,\Ptilde_{X_{1}X_{2}Y}^{\prime})\ge\Gbar_{\delta,n}(P_{X_{1}X_{2}Y})\bigg]\Bigg\}.\label{eq:SUC_MaxSplit} 
\end{align}
For $R_{2}\ge I_{\Ptilde^{\prime}}(X_{2};Y|X_{1})+\delta$,
we have from part 2 of Lemma \ref{lem:COG_SuperExp} that, 
conditioned on $\tilde{\Ac}_{\delta}(\Ptilde_{X_{1}Y})$, 
\begin{equation}
    \tilde{\Xi}_{\yv}(\xvbar_1,\Ptilde_{X_{1}X_{2}Y}^{\prime})\ge M_{2}e^{-n(I_{\Ptilde^{\prime}}(X_{2};Y|X_{1})+\delta)}q^{n}(\Ptilde_{X_{1}X_{2}Y}^{\prime}).
\end{equation}
On the other hand, for $R_{2}\le I_{\Ptilde^{\prime}}(X_{2};Y|X_{1})-\delta$, we have
\begin{align}
      & \PP\Big[\tilde{\Xi}_{\yv}(\xvbar_1,\Ptilde_{X_{1}X_{2}Y}^{\prime})\ge\Gbar_{\delta,n}(P_{X_{1}X_{2}Y})\Big] \\
      & \qquad=\PP\Big[\tilde{\Xi}_{\yv}(\xvbar_1,\Ptilde_{X_{1}X_{2}Y}^{\prime})\ge\Gbar_{\delta,n}(P_{X_{1}X_{2}Y})\,\cap\, N_{\xvbar_{1}\yv}(\Ptilde_{X_{1}X_{2}Y}^{\prime})>0\Big] \label{eq:COG_Xitil_S2} \\
      & \qquad=\PP\Big[\tilde{\Xi}_{\yv}(\xvbar_1,\Ptilde_{X_{1}X_{2}Y}^{\prime})\ge\Gbar_{\delta,n}(P_{X_{1}X_{2}Y})\,\big|\, N_{\xvbar_{1}\yv}(\Ptilde_{X_{1}X_{2}Y}^{\prime})>0\Big]\PP\Big[N_{\xvbar_{1}\yv}(\Ptilde_{X_{1}X_{2}Y}^{\prime})>0\Big] \\
      & \qquad\doteq\PP\Big[\tilde{\Xi}_{\yv}(\xvbar_1,\Ptilde_{X_{1}X_{2}Y}^{\prime})\ge\Gbar_{\delta,n}(P_{X_{1}X_{2}Y})\,\big|\, N_{\xvbar_{1}\yv}(\Ptilde_{X_{1}X_{2}Y}^{\prime})>0\Big]M_{2}e^{-nI_{\Ptilde^{\prime}}(X_{2};Y|X_{1})} \label{eq:COG_Xitil_S4} \\
      & \qquad\,\,\dot{\ge}\,\openone\Big\{ q^{n}(\Ptilde_{X_{1}X_{2}Y}^{\prime})\ge\Gbar_{\delta,n}(P_{X_{1}X_{2}Y})\Big\} M_{2}e^{-nI_{\Ptilde^{\prime}}(X_{2};Y|X_{1})}, \label{eq:COG_Xitil_S5}
\end{align}
where \eqref{eq:COG_Xitil_S2} follows since the event under consideration
is zero unless $N_{\xvbar_{1}\yv}(\Ptilde_{X_{1}X_{2}Y}^{\prime})>0$, \eqref{eq:COG_Xitil_S4}
follows from part 3b of Lemma \ref{lem:COG_SuperExp}, and \eqref{eq:COG_Xitil_S5} 
follows since when $N_{\xvbar_{1}\yv}(\Ptilde_{X_{1}X_{2}Y}^{\prime})$ is positive
it must be at least one.

\subsection*{\bf Step 6: Deduce the Exponent for Fixed $(\xv_1,\xv_2,\yv)$}

We have now handled both cases in \eqref{eq:SUC_MaxSplit}.  Combining them, and substituting the result
into \eqref{eq:COG_PartialTypeExpr}, we obtain
\begin{align}
      & \Phi_1(P_{X_1X_2Y}) \,\,\dot{\ge}\,\max_{\Ptilde_{X_{1}Y} \in \Sicn(Q_{X_1,n},P_{Y})} e^{-nI_{\Ptilde}(X_{1};Y)} \nonumber \\
      & \max\Bigg\{\max_{\substack{\Ptilde_{X_{1}X_{2}Y}^{\prime}\in\Sicn^{\prime}(Q_{X_1X_2,n},\Ptilde_{X_{1}Y})\,:\, \\  R_{2}\ge I_{\Ptilde^{\prime}}(X_{2};Y|X_{1})+\delta}}\openone\Big\{ M_{2}e^{-n(I_{\Ptilde^{\prime}}(X_{2};Y|X_{1})+\delta)}q^{n}(\Ptilde_{X_{1}X_{2}Y}^{\prime})\ge\Gbar_{\delta,n}(P_{X_{1}X_{2}Y})\Big\}, \nonumber \\
      & \qquad\qquad\max_{\substack{\Ptilde_{X_{1}X_{2}Y}^{\prime}\in\Sicn^{\prime}(Q_{X_1X_2,n},\Ptilde_{X_{1}Y})\,:\, \\  R_{2}\le I_{\Ptilde^{\prime}}(X_{2};Y|X_{1})-\delta}}M_{2}e^{-nI_{\Ptilde^{\prime}}(X_{2};Y|X_{1})}\openone\Big\{ q^{n}(\Ptilde_{X_{1}X_{2}Y}^{\prime})\ge\Gbar_{\delta,n}(P_{X_{1}X_{2}Y})\Big\}\Bigg\}. \label{eq:COG_Combined}      
\end{align}
Observe that $\Fbarc(P_{X_1X_2Y})$ in \eqref{eq:SUC_Flower} equals the exponent of $\Gbar_{\delta,n}$ 
in \eqref{eq:COG_Gbar} in the limit as $\delta\to0$ and $n\to\infty$. 
Similarly, the exponent corresponding to the quantity in the first indicator function in 
\eqref{eq:COG_Combined} tends to
\begin{equation}
    \Fibarc(\Ptilde_{X_{1}X_{2}Y}^{\prime},R_{2}) \defeq\EE_{\Ptilde^{\prime}}[\log q(X_{1},X_{2},Y)]+R_{2}-I_{\Ptilde^{\prime}}(X_{2};Y|X_{1}). \label{eq:COG_F1upper} 
\end{equation}
Recalling that $\Phi_1$ is the inner probability in \eqref{eq:COG_Union}, we obtain
the following by taking $\delta\to0$ sufficiently slowly  and using the 
continuity of the underlying terms in the optimizations:
\begin{align}
      & \peibar(\xv_{1},\xv_{2},\yv)\,\,\dot{\ge}\,\max\Bigg\{\max_{(\Ptilde_{X_{1}Y},\Ptilde_{X_{1}X_{2}Y}^{\prime})\in\Tici(P_{X_{1}X_{2}Y},R_{2})}M_{1}e^{-nI_{\Ptilde}(X_{1};Y)},\nonumber \\
      & \qquad\qquad\max_{(\Ptilde_{X_{1}Y},\Ptilde_{X_{1}X_{2}Y}^{\prime})\in\Ticii(P_{X_{1}X_{2}Y},R_{2}) }M_{1}e^{-nI_{\Ptilde}(X_{1};Y)}M_{2}e^{-nI_{\Ptilde^{\prime}}(X_{2};Y|X_{1})}\Bigg\},\label{eq:COG_peFinal}
\end{align}
where
\begin{align}
    & \Tici(P_{X_{1}X_{2}Y},R_{2})\defeq\Big\{(\Ptilde_{X_{1}Y},\Ptilde_{X_{1}X_{2}Y}^{\prime})\,:\,\Ptilde_{X_{1}Y}\in\Sic(Q_{X_1},P_Y), \nonumber \\
    & \Ptilde_{X_{1}X_{2}Y}^{\prime}\in\Sipc(Q_{X_1X_2},\Ptilde_{X_{1}Y}),I_{\Ptilde^{\prime}}(X_{2};Y|X_{1})\le R_{2},\Fibarc(\Ptilde_{X_{1}X_{2}Y}^{\prime},R_{2})\ge\Funderc(P_{X_{1}X_{2}Y},R_{2})\Big\}, \label{eq:COG_SetT1_1} 
\end{align}
\vspace*{-5ex}
\begin{multline}
    \Ticii(P_{X_{1}X_{2}Y},R_{2})\defeq\Big\{(\Ptilde_{X_{1}Y},\Ptilde_{X_{1}X_{2}Y}^{\prime})\,:\,\Ptilde_{X_{1}Y}\in\Sic(Q_{X_1},P_Y), \\
    \Ptilde_{X_{1}X_{2}Y}^{\prime}\in\Sipc(Q_{X_1X_2},\Ptilde_{X_{1}Y}),I_{\Ptilde^{\prime}}(X_{2};Y|X_{1})\ge R_{2},\EE_{\Ptilde^{\prime}}[\log q(X_{1},X_{2},Y)]\ge\Funder(P_{X_{1}X_{2}Y},R_{2})\Big\}, \label{eq:COG_SetT1_2}                                  
\end{multline}
and 
\begin{gather}
    \Sic(Q_{X_1},P_Y)                \defeq \Big\{\Ptilde_{X_{1}X_{2}Y} \in \Pc(\Xc_1 \times \Xc_2 \times \Yc) \,:\, \Ptilde_{X_{1}}=Q_{X_1},\Ptilde_{Y}=P_{Y}\Big\}, \label{eq:COG_SetS1} \\
    \Sipc(Q_{X_1X_2},\Ptilde_{X_{1}Y}) \defeq \Big\{\Ptilde_{X_{1}X_{2}Y}^{\prime} \in \Pc(\Xc_1 \times \Xc_2 \times \Yc) \,:\, \Ptilde_{X_{1}Y}^{\prime}=\Ptilde_{X_{1}Y},\Ptilde_{X_1X_2}^{\prime}=Q_{X_1X_2} \Big\}. \label{eq:COG_SetS1'}
\end{gather} 
Specifically, this follows from the same argument as Step 6 in Section \ref{sec:SUC_MAIN_PROOF}.

\subsection*{\bf Step 7: Deduce the Achievable Rate Region}

Similarly to the achievability proof in Section \ref{sec:SUC_MAIN_PROOF}, the fact that the
joint type of $(\Xv_1,\Xv_2,\Yv)$ approaches $Q_{X_1X_2} \times W$ with probability approaching
one means that we can substitute $P_{X_1X_2Y} = Q_{X_1X_2} \times W$ to obtain the following rate conditions:
\begin{align}
    R_1       &\le \min_{(\Ptilde_{X_{1}Y},\Ptilde_{X_{1}X_{2}Y}^{\prime})\in\Tici(Q_{X_1X_2} \times W,R_{2})} I_{\Ptilde}(X_{1};Y), \label{eq:COG_Final1} \\
    R_1 + R_2 &\le \min_{(\Ptilde_{X_{1}Y},\Ptilde_{X_{1}X_{2}Y}^{\prime})\in\Ticii(Q_{X_1X_2} \times W,R_{2}) } I_{\Ptilde}(X_{1};Y) + I_{\Ptilde^{\prime}}(X_{2};Y|X_{1}). \label{eq:COG_Final2}
\end{align}
The proof of \eqref{eq:COG_R1} is now concluded via the same argument as Step 7 in Section \ref{sec:SUC_MAIN_PROOF},
using the definitions of $\Fbarc$, $\Fibarc$, $\Sic$, $\Sipc$, $\Tici$ and $\Ticii$ to
unite \eqref{eq:COG_Final1}--\eqref{eq:COG_Final2}.  Note that the optimization 
variable $\Ptilde_{X_1Y}$ can be absorbed into $\Ptilde'_{X_1X_2Y}$ due to the constraint 
$\Ptilde'_{X_1Y} = \Ptilde_{X_1Y}$.

\section{Conclusion}

We have obtained error exponents and achievable rates for both the standard and cognitive MAC
using a mismatched multi-letter successive decoding rule.   For the cognitive case, we have
proved ensemble tightness, thus allowing us to conclusively establish that there are cases
in which neither the mismatched successive decoding region nor the mismatched maximum-metric
decoding region \cite{MacMM} dominate each other in the random coding setting.

An immediate direction for further work is to establish the ensemble tightness of the 
achievable rate region for the standard MAC in Theorem \ref{thm:SUC_MainResult}.  A 
more challenging open question is to determine whether either of the \emph{true} 
mismatched capacity regions (rather than just achievable random-coding regions) 
for the two decoding rules contain each other in general.

\appendices

%
%
\section{Behavior of Successive Decoder with $q= W$} \label{sub:SUC_PROOF_ML}

Here we show that a rate pair $(R_1,R_2)$ or 
error exponent $E(R_1,R_2)$ is achievable under
maximum-likelihood (ML) decoding if and only if it is achievable 
under the successive rule in \eqref{eq:SUC_Decoder1}--\eqref{eq:SUC_Decoder2}
with $q(x_1,x_2,y)=W(y|x_1,x_2)$.  This is shown in the same way for the standard MAC and the cognitive MAC,
so we focus on the former.

It suffices to show that, for any fixed codebooks $\Cc_1=\{\xv_1^{(i)}\}_{i=1}^{M_1}$ and
$\Cc_2=\{\xv_2^{(j)}\}_{j=1}^{M_2}$, the error probability under ML decoding
is lower bounded by a constant times the error probability under successive
decoding.  It also suffices to consider the variations of these decoders where
ties are broken as errors, since doing so reduces the error probability 
by at most a factor of two \cite{TwoChannels}. Formally, we consider
the following:
\begin{enumerate}
    \item The ML decoder maximizing $W^n(\yv|\xv_1^{(i)},\xv_2^{(j)})$;
    \item The successive decoder in \eqref{eq:SUC_Decoder1}--\eqref{eq:SUC_Decoder2} with $q=W$;
    \item The genie-aided successive decoder using the true value of $\msg_1$ on the
          second step rather than $\hat{\msg}_1$ \cite{GrantRateSplit}:
          \begin{align}
              \hat{\msg}_{1} &= \arg\max_{i}\sum_{j}W^{n}(\xv_{1}^{(i)},\xv_{2}^{(j)},\yv), \\
              \hat{\msg}_{2} &= \arg\max_{j}W^{n}(\xv_{1}^{(\msg_{1})},\xv_{2}^{(j)},\yv).
          \end{align}        
\end{enumerate}
We denote the probabilities under these decoders by $\PrML[\cdot]$,
$\PrS[\cdot]$ and $\PrGenie[\cdot]$ respectively.  Denoting the 
random message pair by $(\msg_1,\msg_2)$, the resulting estimate by
$(\hat{\msg}_1,\hat{\msg}_2)$, and the output sequence by $\Yv$, we have
\begin{align}
    &\PrML[(\hat{\msg}_1,\hat{\msg}_2) \ne (\msg_1,\msg_2)] \nonumber \\ 
        &~~ \ge \max\bigg\{ \PrML[\hat{\msg}_1 \ne \msg_1], \PrML\bigg[ \bigcup_{j \ne \msg_2}\Big\{W^n(\xv_1^{(\msg_1)},\xv_2^{(j)},\Yv) \ge W^n(\xv_1^{(\msg_1)},\xv_2^{(\msg_2)},\Yv) \Big\} \bigg]\bigg\} \\
        &~~ \ge \max\bigg\{ \PrGenie[\hat{\msg}_1 \ne \msg_1], \PrGenie\bigg[ \bigcup_{j \ne \msg_2}\Big\{W^n(\xv_1^{(\msg_1)},\xv_2^{(j)},\Yv) \ge W^n(\xv_1^{(\msg_1)},\xv_2^{(\msg_2)},\Yv) \Big\} \bigg]\bigg\} \label{eq:SUC_Genie2} \\
        &~~ \ge \frac{1}{2}\PrGenie[(\hat{\msg}_1,\hat{\msg}_2) \ne (\msg_1,\msg_2)] \label{eq:SUC_Genie3} \\
        &~~ = \frac{1}{2}\PrS[(\hat{\msg}_1,\hat{\msg}_2) \ne (\msg_1,\msg_2)] \label{eq:SUC_Genie4},
\end{align} 
where \eqref{eq:SUC_Genie2} follows since the two steps of the genie-aided decoder correspond to minimizing the two terms in the $\max\{\cdot,\cdot\}$, \eqref{eq:SUC_Genie3}
follows by writing $\max\{\PP[A],\PP[B]\} \ge \frac{1}{2}(\PP[A]+\PP[B]) \ge \frac{1}{2}\PP[A \cup B]$,
and \eqref{eq:SUC_Genie4} follows since the overall error probability is unchanged
by the genie \cite{GrantRateSplit}.

\section{Formulations of \eqref{eq:SUC_R1} and \eqref{eq:COG_R1} in Terms of Convex Optimization Problems} \label{sec:SUC_CONVEX}

In this section, we provide an alternative formulation of \eqref{eq:SUC_R1} that is written in terms of convex optimization problems.  We start with the alternative formulation in \eqref{eq:SUC_CondR1}--\eqref{eq:SUC_CondR2}.
We first note that \eqref{eq:SUC_CondR2} holds if and only if
\begin{equation}
    R_{1}\le\min_{\substack{(\Ptilde_{X_{1}X_{2}Y},\Ptilde_{X_{1}X_{2}Y}^{\prime})\in\Tc_{1}^{(2)}(Q_1 \times Q_2 \times W,R_{2})}
        }I_{\Ptilde}(X_{1};X_{2},Y)+\big[I_{\Ptilde^{\prime}}(X_{2};X_{1},Y)-R_{2}\big]^{+},\label{eq:SUC_CondR2a}
\end{equation}
since the argument to the $[\cdot]^{+}$ is always non-negative
due to the constraint $I_{\Ptilde^{\prime}}(X_{2};X_{1},Y)\ge R_{2}$.
Next, we claim that when combining \eqref{eq:SUC_CondR1} and \eqref{eq:SUC_CondR2a},
the rate region is unchanged if the constraint $I_{\Ptilde^{\prime}}(X_{2};X_{1},Y)\ge R_{2}$
is omitted from \eqref{eq:SUC_CondR2a}.  This is seen by noting that
whenever $I_{\Ptilde^{\prime}}(X_{2};X_{1},Y)<R_{2}$, the objective
in \eqref{eq:SUC_CondR2a} coincides with that of \eqref{eq:SUC_CondR1},
whereas the latter has a less restrictive constraint since $\Fbar_{1}>\Fbar_{2}$
(see \eqref{eq:SUC_F1upper}--\eqref{eq:SUC_F2upper}).

We now deal with the non-concavity of the functions $\Fbar_{1}$ and $\Fbar_{2}$
appearing in the sets $\Tc_1^{(1)}$ and $\Tc_1^{(2)}$.
Using the identity
\begin{equation}
    \min_{x\le\max\{a,b\}}f(x)=\min\Big\{\min_{x\le a}f(x),\min_{x\le b}f(x)\Big\},
\end{equation}
we obtain the following rate conditions from \eqref{eq:SUC_CondR1}
and \eqref{eq:SUC_CondR2a}: 
\begin{align}
    R_{1} & \le\min_{\substack{(\Ptilde_{X_{1}X_{2}Y},\Ptilde_{X_{1}X_{2}Y}^{\prime})\in\Tc_{1}^{(1,1)}(Q_1 \times Q_2 \times W,R_{2})} }I_{\Ptilde}(X_{1};X_{2},Y), \label{eq:SUC_CondR1a'}\\
    R_{1} & \le\min_{\substack{(\Ptilde_{X_{1}X_{2}Y},\Ptilde_{X_{1}X_{2}Y}^{\prime})\in\Tc_{1}^{(1,2)}(Q_1 \times Q_2 \times W,R_{2})} }I_{\Ptilde}(X_{1};X_{2},Y), \label{eq:SUC_CondR1b'}\\
    R_{1} & \le\min_{\substack{(\Ptilde_{X_{1}X_{2}Y},\Ptilde_{X_{1}X_{2}Y}^{\prime})\in\Tc_{1}^{(2,1)}(Q_1 \times Q_2 \times W,R_{2})} }I_{\Ptilde}(X_{1};X_{2},Y)+\big[I_{\Ptilde^{\prime}}(X_{2};X_{1},Y)-R_{2}\big]^{+}, \label{eq:SUC_CondR1c'}\\
    R_{1} & \le\min_{\substack{(\Ptilde_{X_{1}X_{2}Y},\Ptilde_{X_{1}X_{2}Y}^{\prime})\in\Tc_{1}^{(2,2)}(Q_1 \times Q_2 \times W,R_{2})} }I_{\Ptilde}(X_{1};X_{2},Y)+\big[I_{\Ptilde^{\prime}}(X_{2};X_{1},Y)-R_{2}\big]^{+},\label{eq:SUC_CondR1d'}
\end{align}
where
\begin{align}
    & \Tc_{1}^{(1,1)}(P_{X_{1}X_{2}Y},R_{2})\defeq\Big\{(\Ptilde_{X_{1}X_{2}Y},\Ptilde_{X_{1}X_{2}Y}^{\prime}):\Ptilde_{X_{1}X_{2}Y}\in\Sc_{1}(Q_{1},P_{X_{2}Y}),  \nonumber \\
    & \Ptilde_{X_{1}X_{2}Y}^{\prime}\in\Sc_{1}^{\prime}(Q_{2},\Ptilde_{X_{1}Y}), I_{\Ptilde^{\prime}}(X_{2};X_{1},Y)\le R_{2},\EE_{\Ptilde}[\log q(X_{1},X_{2},Y)]\ge\Funder(P_{X_{1}X_{2}Y},R_{2})\Big\}, \label{eq:SUC_SetT1_11}
\end{align}
\vspace*{-5ex}
\begin{align}
     \Tc_{1}^{(1,2)}(P_{X_{1}X_{2}Y},R_{2})& \defeq\Big\{(\Ptilde_{X_{1}X_{2}Y},\Ptilde_{X_{1}X_{2}Y}^{\prime}):\Ptilde_{X_{1}X_{2}Y}\in\Sc_{1}(Q_{1},P_{X_{2}Y}),  \nonumber \\
    & \hspace*{-4ex} \Ptilde_{X_{1}X_{2}Y}^{\prime}\in\Sc_{1}^{\prime}(Q_{2},\Ptilde_{X_{1}Y}), I_{\Ptilde^{\prime}}(X_{2};X_{1},Y)\le R_{2}, \nonumber \\ 
    & \hspace*{-4ex} \EE_{\Ptilde^{\prime}}[\log q(X_{1},X_{2},Y)]+R_{2}-I_{\Ptilde^{\prime}}(X_{2};X_{1},Y)\ge\Funder(P_{X_{1}X_{2}Y},R_{2})\Big\}, \label{eq:SUC_SetT1_12} 
\end{align}
\vspace*{-5ex}
\begin{multline}
    \Tc_{1}^{(2,1)}(P_{X_{1}X_{2}Y},R_{2})\defeq\Big\{(\Ptilde_{X_{1}X_{2}Y},\Ptilde_{X_{1}X_{2}Y}^{\prime}):\Ptilde_{X_{1}X_{2}Y}\in\Sc_{1}(Q_{1},P_{X_{2}Y}), \\
    \Ptilde_{X_{1}X_{2}Y}^{\prime}\in\Sc_{1}^{\prime}(Q_{2},\Ptilde_{X_{1}Y}),\EE_{\Ptilde}[\log q(X_{1},X_{2},Y)]\ge\Funder(P_{X_{1}X_{2}Y},R_{2})\Big\}, \label{eq:SUC_SetT1_21}                      
\end{multline}
\vspace*{-5ex}
\begin{multline}
    \Tc_{1}^{(2,2)}(P_{X_{1}X_{2}Y},R_{2})\defeq\Big\{(\Ptilde_{X_{1}X_{2}Y},\Ptilde_{X_{1}X_{2}Y}^{\prime}):\Ptilde_{X_{1}X_{2}Y}\in\Sc_{1}(Q_{1},P_{X_{2}Y}), \\
    \Ptilde_{X_{1}X_{2}Y}^{\prime}\in\Sc_{1}^{\prime}(Q_{2},\Ptilde_{X_{1}Y}),\EE_{\Ptilde^{\prime}}[\log q(X_{1},X_{2},Y)]\ge\Funder(P_{X_{1}X_{2}Y},R_{2})\Big\}.\label{eq:SUC_SetT1_22}                       
\end{multline}
These are obtained from $\Tc^{(k)}$ ($k=1,2$) by keeping only one term in the definition of
$\Fbar_k$ (see \eqref{eq:SUC_F1upper}--\eqref{eq:SUC_F2upper}), and by removing the
constraint $I_{\Ptilde'}(X_2;X_1,Y) \ge R_2$ when $k=2$ in accordance with the discussion following \eqref{eq:SUC_CondR2a}.

The variable $\Ptilde_{X_{1}X_{2}Y}^{\prime}$ can be removed
from both \eqref{eq:SUC_CondR1a'} and \eqref{eq:SUC_CondR1c'}, since
in each case the choice $\Ptilde_{X_{1}X_{2}Y}^{\prime}(x_{1},x_{2},y)=Q_{2}(x_{2})\Ptilde_{X_{1}Y}(x_{1},y)$
is feasible and yields $I_{\Ptilde'}(X_2;X_1,Y) = 0$.
It follows that \eqref{eq:SUC_CondR1a'} and \eqref{eq:SUC_CondR1c'}
yield the same value, and we conclude that \eqref{eq:SUC_R1} can
equivalently be expressed in terms of three conditions: \eqref{eq:SUC_CondR1b'},
\eqref{eq:SUC_CondR1d'}, and 
\begin{align}
    R_{1} & \le\min_{\Ptilde_{X_{1}X_{2}Y}\in\Tc_{1}^{(1,1^{\prime})}(Q_1 \times Q_2 \times W,R_{2})}I_{\Ptilde}(X_{1};X_{2},Y),\label{eq:SUC_CondR1a''} 
\end{align}
where the set
\begin{equation}
    \Tc_{1}^{(1,1')}(P_{X_{1}X_{2}Y},R_{2}) \defeq \Big\{ \Ptilde_{X_{1}X_{2}Y}\in\Sc_{1}(Q_{1},P_{X_{2}Y}) \,:\, \EE_{\Ptilde}[\log q(X_{1},X_{2},Y)]\ge\Funder(P_{X_{1}X_{2}Y},R_{2})\Big\}
\end{equation}
is obtained by eliminating $\Ptilde'_{X_1X_2Y}$ from either \eqref{eq:SUC_SetT1_11}
or \eqref{eq:SUC_SetT1_21}. These three conditions are all written as convex optimization
problems, as desired.

Starting with \eqref{eq:COG_Final1}--\eqref{eq:COG_Final2}, one can 
follow a (a simplified version of) the above arguments for 
the cognitive MAC to show that \eqref{eq:COG_R1} holds if and only if
\begin{align}
    R_1 &\le \min_{(\Ptilde_{X_{1}Y},\Ptilde_{X_{1}X_{2}Y}^{\prime})\in\Tici(Q_{X_{1}X_{2}}\times W,R_{2})} I_{\Ptilde}(X_{1};Y),  \label{eq:COG_Final3} \\
    R_1 &\le \min_{(\Ptilde_{X_{1}Y},\Ptilde_{X_{1}X_{2}Y}^{\prime})\in\Ticiip(Q_{X_{1}X_{2}}\times W,R_{2}) } I_{\Ptilde}(X_{1};Y) + \big[I_{\Ptilde^{\prime}}(X_{2};Y|X_{1}) - R_2\big]^+. \label{eq:COG_Final4}
\end{align}
where
\begin{multline}
    \Ticiip(P_{X_{1}X_{2}Y},R_{2})\defeq\Big\{(\Ptilde_{X_{1}Y},\Ptilde_{X_{1}X_{2}Y}^{\prime})\,:\,\Ptilde_{X_{1}Y}\in\Sic(Q_{X_1},P_Y), \\
    \Ptilde_{X_{1}X_{2}Y}^{\prime}\in\Sipc(Q_{X_1X_2},\Ptilde_{X_{1}Y}),\EE_{\Ptilde^{\prime}}[\log q(X_{1},X_{2},Y)]\ge\Funder(P_{X_{1}X_{2}Y},R_{2})\Big\}, \label{eq:COG_SetT1_2'}                                                                                                                                                                                                                                                                                      
\end{multline}
and where $\Tici$, $\Sic$ and $\Sipc$ are defined in \eqref{eq:COG_SetT1_1}--\eqref{eq:COG_SetS1'}.

\bibliographystyle{IEEEtran}
\bibliography{12-Paper,18-MultiUser,18-SingleUser,35-Other}

\end{document}

%% file: preamble.tex
\newcommand{\overbar}[1]{\mkern 1.25mu\overline{\mkern-1.25mu#1\mkern-0.25mu}\mkern 0.25mu}
\newcommand{\overbarr}[1]{\mkern 3.5mu\overline{\mkern-3.5mu#1\mkern-0.75mu}\mkern 0.75mu}

\newcommand{\openone}{\mathds{1}}

\newcommand{\Itilde}{\tilde{I}}

\newcommand{\pebar}{\overbar{p}_{\mathrm{e}}}

\newcommand{\pe}{p_{\mathrm{e}}}

\newcommand{\msg}{\mathsf{m}}

\newcommand{\Fbar}{\overbar{F}}
\newcommand{\Funder}{\underline{F}}
\newcommand{\Gbar}{\overbar{G}}
\newcommand{\Gunder}{\underline{G}}
\newcommand{\Fbarc}{\overbar{F}_{\mathrm{c}}}
\newcommand{\Fibarc}{\overbar{F}_{1\mathrm{c}}}
\newcommand{\Funderc}{\underline{F}_{\mathrm{c}}}

\newcommand{\Tic}{\mathcal{T}_{1\mathrm{c}}}
\newcommand{\Tici}{\mathcal{T}_{1\mathrm{c}}^{(1)}}
\newcommand{\Ticii}{\mathcal{T}_{1\mathrm{c}}^{(2)}}
\newcommand{\Ticiip}{\mathcal{T}_{1\mathrm{c}}^{(2')}}
\newcommand{\Tipc}{\mathcal{T}^{\prime}_{1\mathrm{c}}}
\newcommand{\Tiic}{\mathcal{T}_{2\mathrm{c}}}
\newcommand{\Ioc}{I_{0\mathrm{c}}}
\newcommand{\Sicn}{\mathcal{S}_{1\mathrm{c},n}}
\newcommand{\Sic}{\mathcal{S}_{1\mathrm{c}}}
\newcommand{\Sipc}{\mathcal{S}_{1\mathrm{c}}^{\prime}}

\newcommand{\Ptilde}{\widetilde{P}}

\newcommand{\xvbar}{\overbar{\boldsymbol{x}}}

\newcommand{\xv}{\boldsymbol{x}}

\newcommand{\Xvbar}{\overbarr{\boldsymbol{X}}}

\newcommand{\Xv}{\boldsymbol{X}}
\newcommand{\yv}{\boldsymbol{y}}
\newcommand{\Yv}{\boldsymbol{Y}}

\newcommand{\Ac}{\mathcal{A}}
\newcommand{\Bc}{\mathcal{B}}
\newcommand{\Cc}{\mathcal{C}}

\newcommand{\Ec}{\mathcal{E}}

\newcommand{\Pc}{\mathcal{P}}
\newcommand{\Sc}{\mathcal{S}}
\newcommand{\Tc}{\mathcal{T}}

\newcommand{\Xc}{\mathcal{X}}
\newcommand{\Yc}{\mathcal{Y}}

\newcommand{\EE}{\mathbb{E}}
\newcommand{\PP}{\mathbb{P}}

\newcommand{\defeq}{\triangleq}

\newcommand{\peibar}{\overbar{p}_{\mathrm{e},1}}
\newcommand{\peiibar}{\overbar{p}_{\mathrm{e},2}}

\newcommand{\PrML}{\PP^{(\mathrm{ML})}}
\newcommand{\PrS}{\PP^{(\mathrm{S})}}
\newcommand{\PrGenie}{\PP^{(\mathrm{Genie})}}